\documentclass[a4paper,aps,twocolumn,nobibnotes,pra]{revtex4-1}

\usepackage[T1]{fontenc}
\usepackage[utf8]{inputenc}
\usepackage{lmodern}
\usepackage{graphicx}
\usepackage{dcolumn}
\usepackage{bm}
\usepackage{amsmath}
\usepackage{amstext}
\usepackage{amssymb}
\usepackage{amsthm}
\usepackage{dsfont}
\usepackage{color}
\usepackage[caption=false]{subfig}
\usepackage{tabularx}
\usepackage{booktabs}
\usepackage[table]{xcolor}
\usepackage{colortbl}
\definecolor{lightgray}{rgb}{0.85,0.85,0.85}

\newcommand{\tabsep}{\vrule width .7pt}

\newcommand{\eq}[1]{Eq.~(\ref{eq:#1})}

\newcommand{\be}{\begin{equation}}
\newcommand{\ee}{\end{equation}}

\newcommand{\bea}{\begin{eqnarray}}
\newcommand{\eea}{\end{eqnarray}}

\newcommand{\mean}[1]{\ensuremath{\langle{#1}\rangle}}

\renewcommand{\rho}{\varrho}
\renewcommand{\tilde}{\widetilde}

\theoremstyle{definition}
\newtheorem{theorem}{Theorem}

\newcommand{\expect}[1]{\langle #1 \rangle}

\newcommand{\conv}{\ensuremath \text{conv}}

\newcommand{\rank}{\ensuremath \text{rank}}
\renewcommand{\vec}[1]{\ensuremath \bm{#1}}

\begin{document}

\title{Finding optimal Bell inequalities 
using the cone-projection technique}

\author{Fabian Bernards}
\author{Otfried Gühne}%
\affiliation{
Naturwissenschaftlich-Technische Fakultät, 
Universität Siegen, 
Walter-Flex-Straße 3, 
57068 Siegen, Germany}

\date{\today}

\begin{abstract}
Bell inequalities are relevant for many problems in quantum 
information science, but finding them for many particles
is computationally hard. Recently, a computationally
feasible method called cone-projection technique has been developed to find all optimal 
Bell inequalities under some constraints, which may be given 
by some symmetry or other linear conditions. In this paper
we extend this work in several directions. We use the method
to generalize the I4422 inequality to three particles  
and a so-called GYNI inequality to four particles. Additionally, 
we find Bell inequalities for three particles that generalize 
the I3322 inequality and the CHSH inequality at the same time. 
We discuss the obtained inequalities in some detail and 
characterize their violation in quantum mechanics. 
\end{abstract}

\maketitle

{\section{ Introduction}

Bell inequalities are a tool to test whether experimental 
data is compatible with a local hidden variable (LHV) model 
or not \cite{brunnerbellreview, scarani2019}. Quantum mechanics 
predicts that Bell inequalities can be violated and this has been 
verified experimentally in an unambiguous manner \cite{shalm2015,hensen2015,giustina2015, rosenfeld2017}. The violation of a Bell inequality requires
the presence of entanglement, so Bell inequalities provide
the possibility to detect entanglement in a device-independent
manner. This can be extended to certify certain quantum states 
and measurements, by a procedure called self-testing \cite{supic2019}.
Moreover, in quantum cryptography Bell inequality violations can
certify secret-key rates \cite{holz2019}. All these applications 
exemplify that the violation of a Bell inequality can be useful.
Even beyond that, some Bell inequalities can be interesting if 
they cannot be violated by quantum mechanics. Indeed, in this 
case they may distinguish quantum mechanics from more general 
non-signaling theories \cite{almeida2010}. The wide range of 
applications of Bell inequalities beyond their original purpose 
of refuting LHV theories makes it desirable to find Bell 
inequalities with interesting properties.

This, however, is an arduous endeavour as the correlations
stemming from LHV models form a high-dimensional polytope 
and identifying interesting Bell inequalities amounts to
finding the corresponding facets of this local polytope 
\cite{peres1999}. Finding the facets of a high-dimensional 
polytope given its extreme points is known to be a hard task 
and its computational complexity quickly renders it intractable 
as the number of parties, measurements per party or outcomes 
per measurement rise \cite{pitowsky1991}.

Furthermore, the number of facet-defining Bell inequalities increases 
rapidly, making it difficult to identify the inequalities of interest.
To give an example, there are only eight non-trivial Bell inequalities 
all of which are versions of the famous Clauser-Horne-Shimony-Holt (CHSH) 
inequality \cite{clauser1969, clauser1970erratum} in the scenario with 
two parties and two different dichotomic measurements per party each.
For the same number of parties, but three dichotomic measurements per 
party, there are already 648 non-trivial facet-defining inequalities, 
72 of which are of the CHSH type and 576 are variations of the so-called 
I3322 inequality. The latter inequality was first identified by Froissart \cite{froissart1981} and later independently by \'Sliwa \cite{sliwa2003} 
and Collins and Gisin \cite{collins2004}. If one increases the number of 
parties, for three parties with two dichotomic measurements each, there 
are 53856 facet-defining inequalities and 46 inequivalent classes of Bell
inequalities \cite{sliwa2003}. Going to even more complex scenarios it 
is impossible so far to compute all the facets. 

Luckily however, in practice we rarely need a complete characterization 
of the local polytope. Instead, we seek Bell inequalities with properties 
that are suitable for a specific purpose. In a recent work \cite{bernards2020}
we proposed a general method for this problem. Henceforth we refer to this
method as the cone-projection technique (CPT). The CPT can be used
to find all optimal Bell inequalities obeying some affine constraints. 
Specifically, one important task that can be addressed with the CPT is finding 
generalizations of a Bell inequality.

In this paper we present a an extended description of the CPT and use
it then to study several scenarios. First, we present a detailed analysis 
of the properties of three-particle generalizations of the I3322 inequality 
previously found \cite{bernards2020}. Second, we study generalization of the 
so-called I4422 inequality \cite{collins2004} to three particles. Third, we 
find and investigate three-particle Bell inequalities that are generalizations 
of the CHSH inequality and the I3322 inequality at the same time. Finally, 
we study generalizations of three-particle Guess-Your-Neighbors-Input (GYNI) 
inequality \cite{almeida2010} to four particles. 

Before embarking into the technical description of the CPT and the 
examples, it may be useful to explain the notion of a generalization 
of a Bell inequality to more particles. Let us consider an example.
Mermin's inequality, which reads 
\begin{align}
\mean{A_1 B_1 C_2} +
\mean{A_1 B_2 C_1 } +
\mean{A_2 B_1 C_1} -
\mean{A_2 B_2 C_2} 
\le 2
\label{eq-mermin}
\end{align}
can be seen as a generalization of the CHSH inequality
\begin{align}
\mean{A_1 B_1} +
\mean{A_1 B_2} +
\mean{A_2 B_1} -
\mean{A_2 B_2} 
\le 2.
\label{eq-chsh}
\end{align}
This should be understood in the following way: 
If Charlie fixes his outcomes $\pm 1$ for the measurements 
$C_1$ and $C_2$ deterministically (instead of performing 
an actual measurement) then Alice and Bob perform essentially 
a CHSH test on the two-body marginal of the shared three-party 
state. Specifically, if we insert $C_1 = C_2 = 1$ into Mermin's
inequality, it reduces to the CHSH inequality. We therefore 
say that Mermin's inequality is reducible to the CHSH inequality. 
Moreover, both the CHSH inequality and the Mermin inequality are 
define facets on the local polytope. Both properties combined 
make the Mermin inequality a generalization of the CHSH inequality. 
In some situations, these two conditions are too weak to find all 
the Bell inequalities that satisfy them, then one may impose some 
symmetry constraints in addition; this can also be handled
by the CPT. In the above example one may for instance 
impose invariance under arbitrary permutations of parties. 

\section{Description of the cone-projection technique (CPT)}
\label{method}
In this section, we provide a detailed description of the method presented
in Ref.~\cite{bernards2020}. The reader who is already familiar with it may
directly skip to the next section.


\begin{figure}[t]
\includegraphics[width=.3\textwidth]{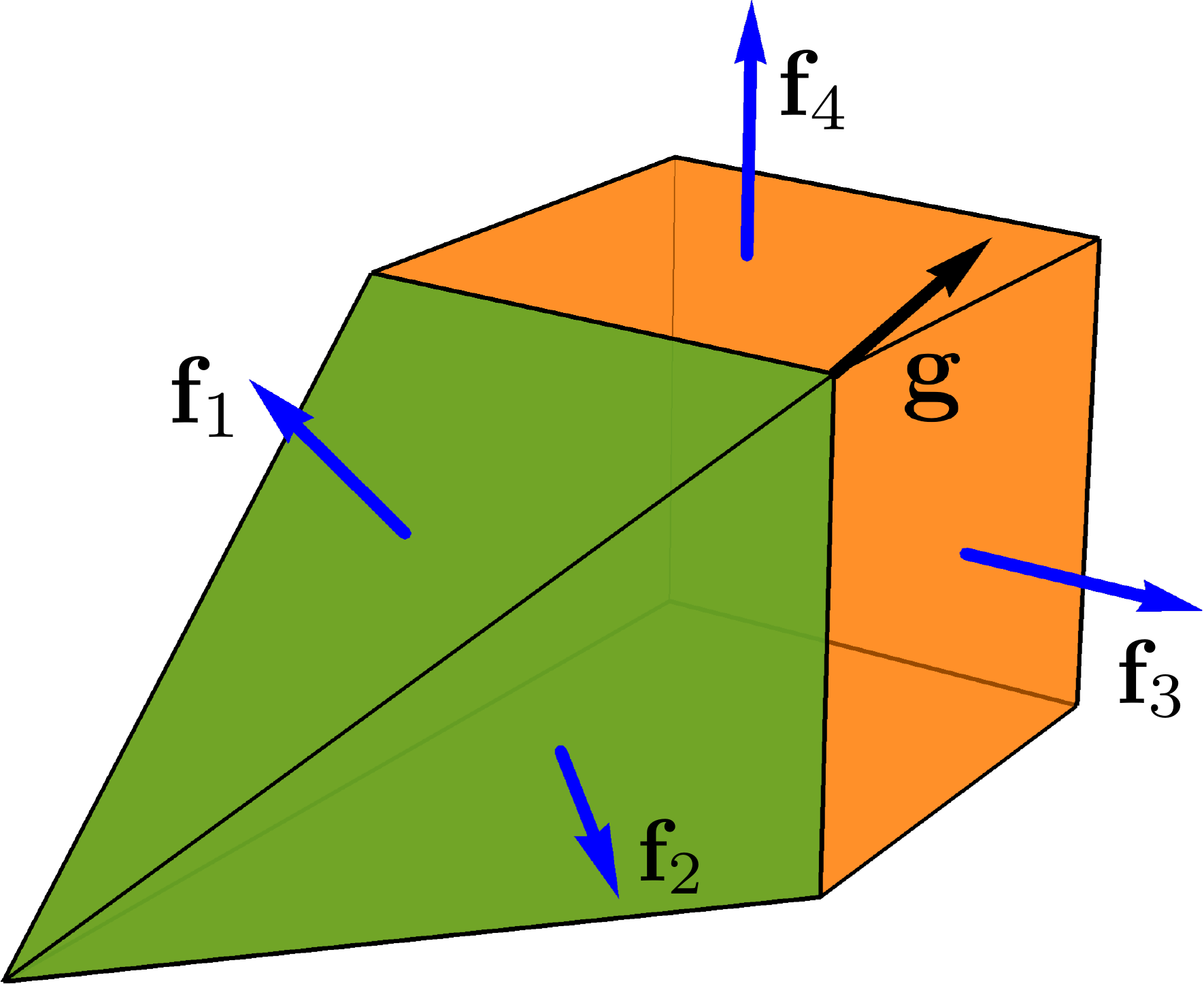}
\caption{A polytope in three-dimensional space is defined by its two-dimensional
surfaces, the so-called facets. Here, the normal vectors of the facets are
depicted as blue vectors. The vector $\vec g$ (black arrow) may be used to define 
a condition on the facets, be requiring that the facet normal vectors $\vec
f_i$
must enclose some absolute angle $\alpha$. For instance, this angle may be chosen such
that the facets drawn in green meet the constraints while the others drawn in orange 
do not. }
\label{fig-1}
\end{figure}

\subsection{Overview and motivation}
We wish to find a method that allows to compute facets of a polytope 
which is defined by its vertices, where the facet normal vectors obey 
some linear constraints. This can be interpreted geometrically as the 
condition that the facet normal vectors must have a fixed inner product
with some vector $\vec g$ that represents the condition, see also 
Fig.~\ref{fig-1}. A naive method to tackle this problem is to compute
first all facets and then to find out which of them obey the constraint. 
For the problems we consider, however, this is not feasible, as it is
already impossible to compute all the facets.

Note that constraints of the considered type include two important 
special cases of constraints that are important in this work: 
The condition that some given vertex of the polytope should lie 
on the desired facet and the condition that the normal vector 
should be symmetric under some linear transformation. In the 
first case, the position vector of the vertex in question takes 
the role of $\vec g$. To understand this, consider a Bell inequality
$\vec x^T \vec b \le -\beta$ with facet normal vector $\vec b$ that should
hold for all classical behaviors $\vec x$. Demanding that some vertex $\vec g$
lies on the facet means that $\vec g^T \vec b = -\beta$, which is an affine
constraint. In the second case all points that obey 
the symmetry lie in a plane and the normal vectors of the desired 
facets have to lie in this plane as well. This is illustrated in Fig.~\ref{fig-symm}.

\begin{figure}[t]
\includegraphics[width=.3\textwidth]{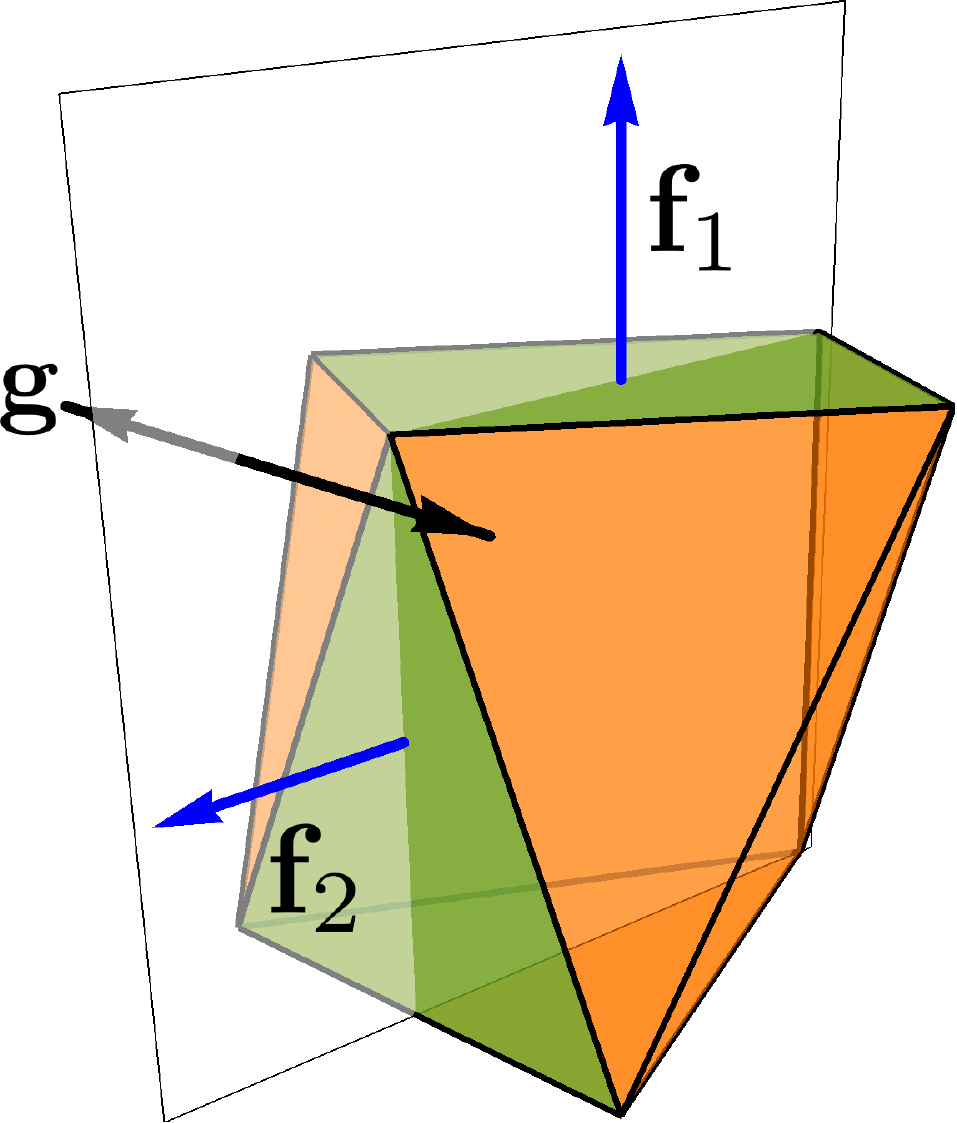}
\caption{This example shows how symmetries can be formulated 
as affine constraints. We seek facets with normal vectors that 
are invariant under reflection on a plane (white, half-transparent). 
This is satisfied by the normal vectors $\vec f_1, \vec f_2$ (blue) of two of the facets 
(green). The condition is equivalent to demanding that the normal 
vectors be perpendicular to the normal vectors (black arrows) of 
the mirror plane.}
\label{fig-symm}
\end{figure}

\subsection{Description of convex cones and polytopes}
For formulating the method in a precise manner, we 
need some more terminology \cite{ziegler}. In detail, 
we need to define the notions of polytopes, cones and 
polyhedra, as well as their so-called V-representation
and H-representation. Finally, we define facets and 
discuss the dimensions of these objects.

The two most fundamental concepts in this context are 
{\it conic combinations} and {\it convex combinations}. 
Conic combinations are linear combinations with positive 
coefficients. Convex combinations are linear combinations
with positive coefficients where the coefficients sum up 
to one. Given a set $V$ of vectors, the set of all conic 
(convex) combinations of the elements in $V$ is called the 
conic (convex) hull of $V$. Also, the conic (convex) hull 
of $V$ is said to be generated by $V$ under conic (convex) 
combinations. If $V$ is finite, its convex hull is called 
finitely generated. 

These notions give rise to the two main objects that are
studied in the following, {\it convex cones} and {\it convex 
polytopes}. The first 
are assumed to be finitely generated under conic combinations 
and the second being finitely generated under convex combinations. 
The elements of $V$ are called {\it vertices} in the case of  
polytopes. For cones, they are called {\it rays}.

The notion of polytopes and cones can be unified using
the concept of {\it polyhedra}. First, given two sets $A$ 
and $B$, one can define their so-called Minkowski sum 
as $A+B = \{a+b | a \in A, b \in B\}$. The sum of a cone 
and a polytope is called a polyhedron. The sets of rays and 
vertices that generate the polyhedron are called the
{\it V-representation} of the polyhedron. Fig.~(\ref{fig-minkowski})
illustrates the Minkowski sum of a line segment and a ray.

\begin{figure}[t]
\includegraphics[width=.7\columnwidth]{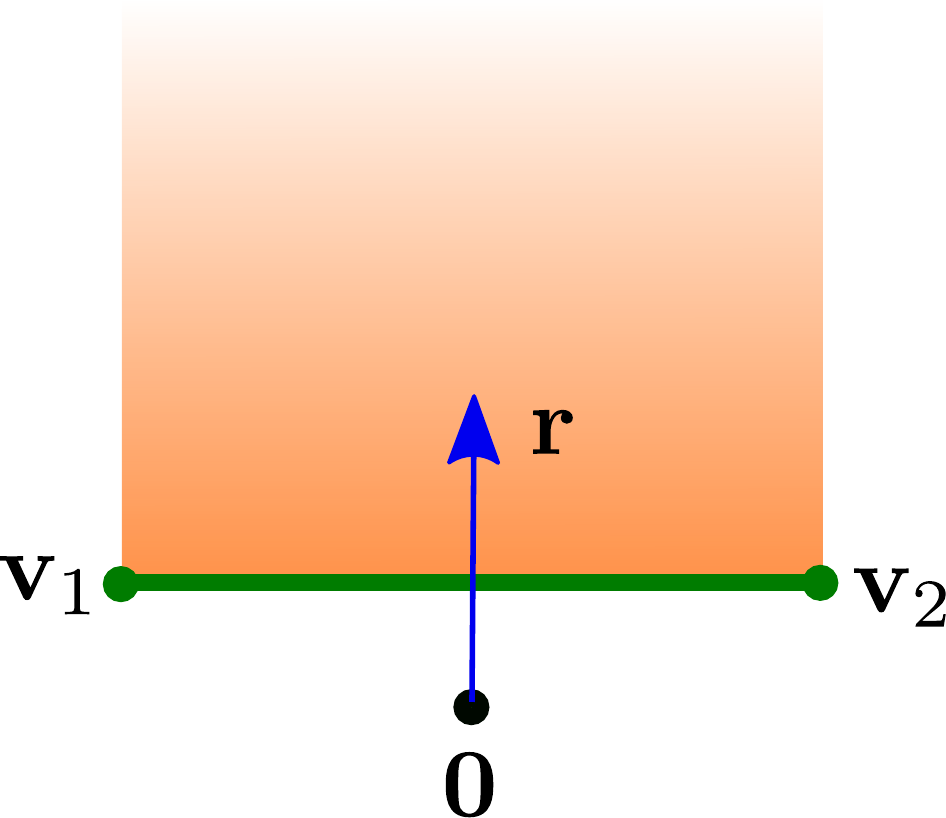}
\caption{The vertices $\vec v_1$, $\vec v_2$ (green points)
define a 1-dimensional polytope (green bold line). The Minkowski sum
of this polytope and the ray $\vec r$ (blue) is a polyhedron (orange), which
extends infinitely to the top. The polyhedron has three facets, which 
are the
lines that confine it to the bottom, the left, and the right.}
\label{fig-minkowski}
\end{figure}

That said, there is a second important representation of 
polyhedra, the {\it H-representation.} Any polyhedron can
be viewed as the intersection of a finite number of half-spaces, 
defined by affine inequalities. The minimal set of half-spaces 
whose intersection is the polyhedron is its H-representation. 

Let us finally introduce the concept of a facet. With every of the 
half-spaces in the H-representation, we can associate the hyperplane 
that bounds it. The intersection between such a hyperplane and 
the polyhedron is called a {\it facet} of the polyhedron. If an 
inequality defines a half-space such that its bounding hyperplane 
contains a facet of a given polyhedron, this inequality is called 
{\it facet-defining} with respect to the polyhedron. From a physical 
perspective, such facet-defining inequalities are the most interesting
Bell inequalities. 

If a polyhedron $P$ has dimension $D$, its facets are the $D-1$ 
dimensional polyhedra that together form the boundary of $P$. 
Intersections of facets are called faces. In the following, 
an $n$-face is a face that has an (affine) dimension of
$n$. Accordingly, in the case of a $D$ dimensional polytope, 
$0$-faces are vertices, $1$-faces are edges and $(D-1)$-faces 
are facets.

\subsection{Targeted search for facets}
Now we are in the position to describe the problem that the cone-projection
technique described in
Ref.~\cite{bernards2020} solves. We consider the situation where 
a convex polytope $P$ is given in its V-representation and we aim to find 
all facets of $P$ that satisfy some linear conditions. The CPT works 
for the case that the conditions are affine 
equality constraints on the coefficients of the facet defining 
inequality. Note that the CPT can be applied to any convex polyhedron, but in
practice we are interested in the local polytope, where the vertices 
are given by deterministic assignments of measurement results in a
LHV model.

\begin{figure}[t]
  \subfloat[\label{sfig:1a}]{\includegraphics[width=0.17\textwidth]{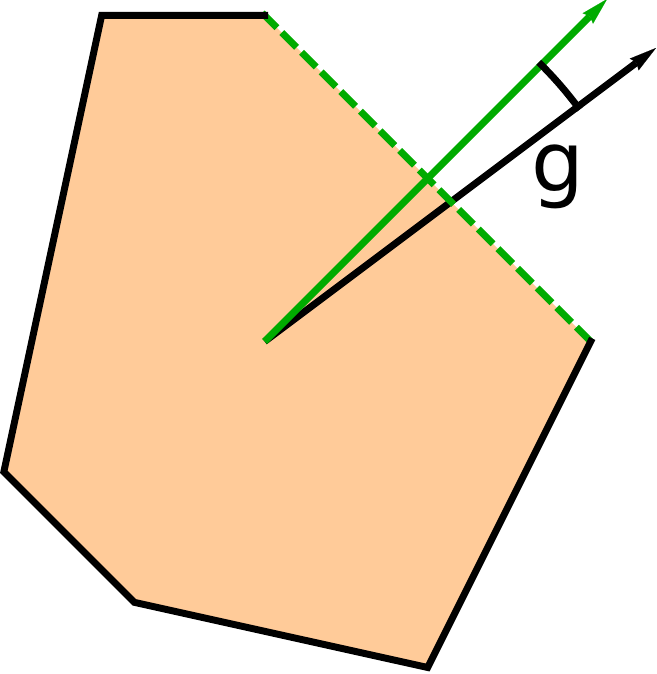}}
  \hspace{0.3cm}
  \subfloat[\label{sfig:1b}]{\includegraphics[width=0.25\textwidth]{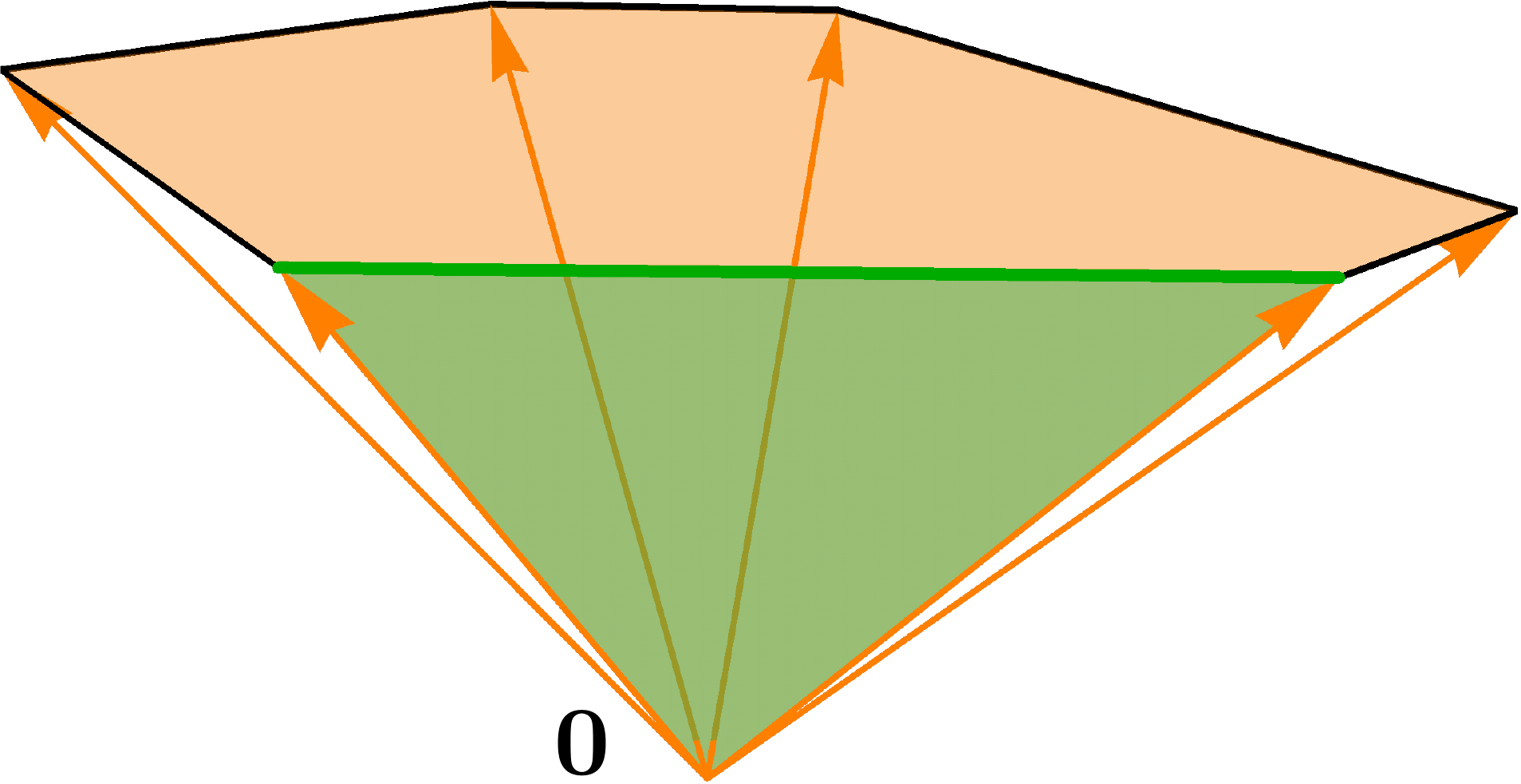}}
  
   \subfloat[\label{sfig:1c}]{\includegraphics[width=0.22\textwidth]{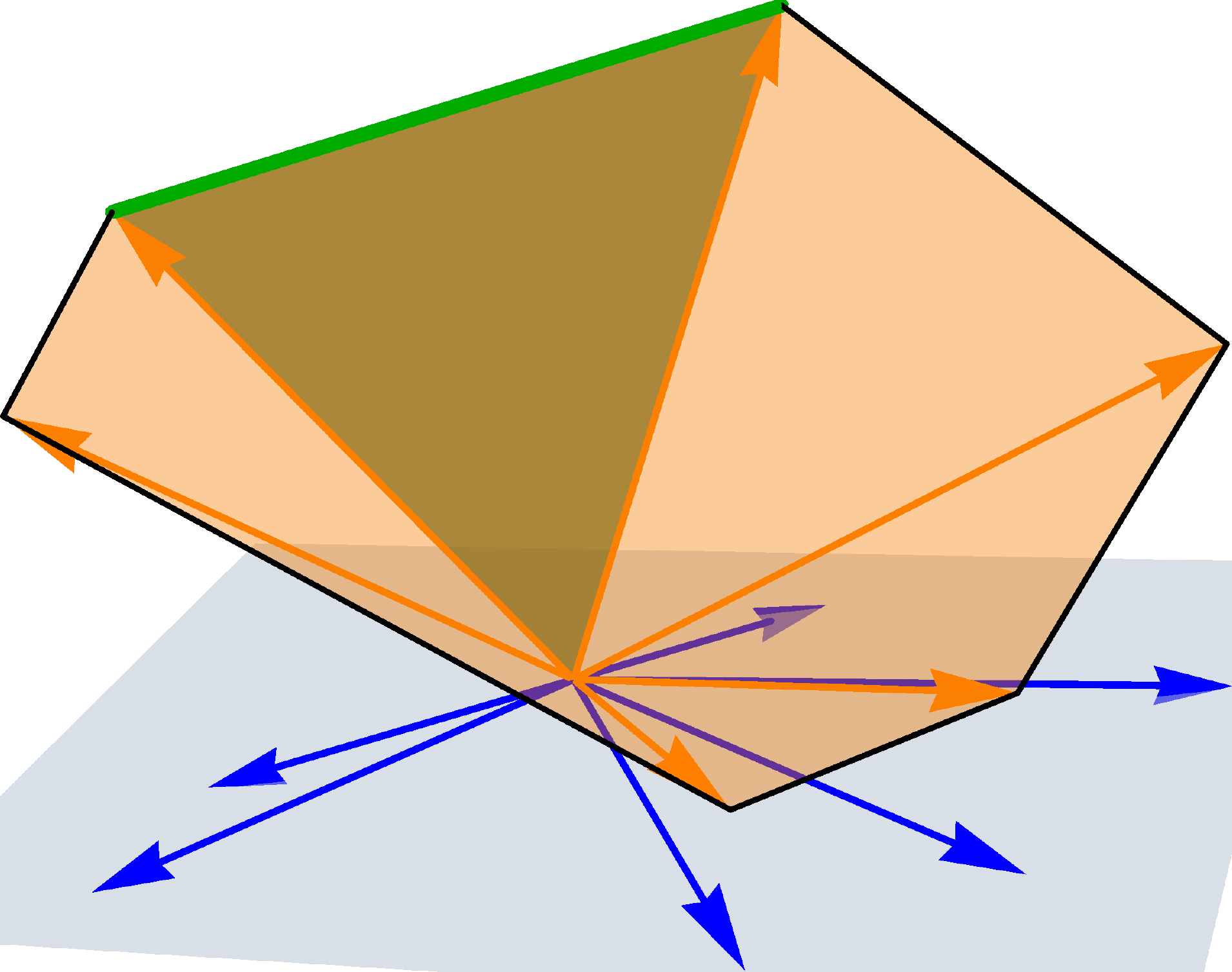}}
  \hspace{0.3cm}
  \subfloat[\label{sfig:1d}]{\includegraphics[width=0.21\textwidth]{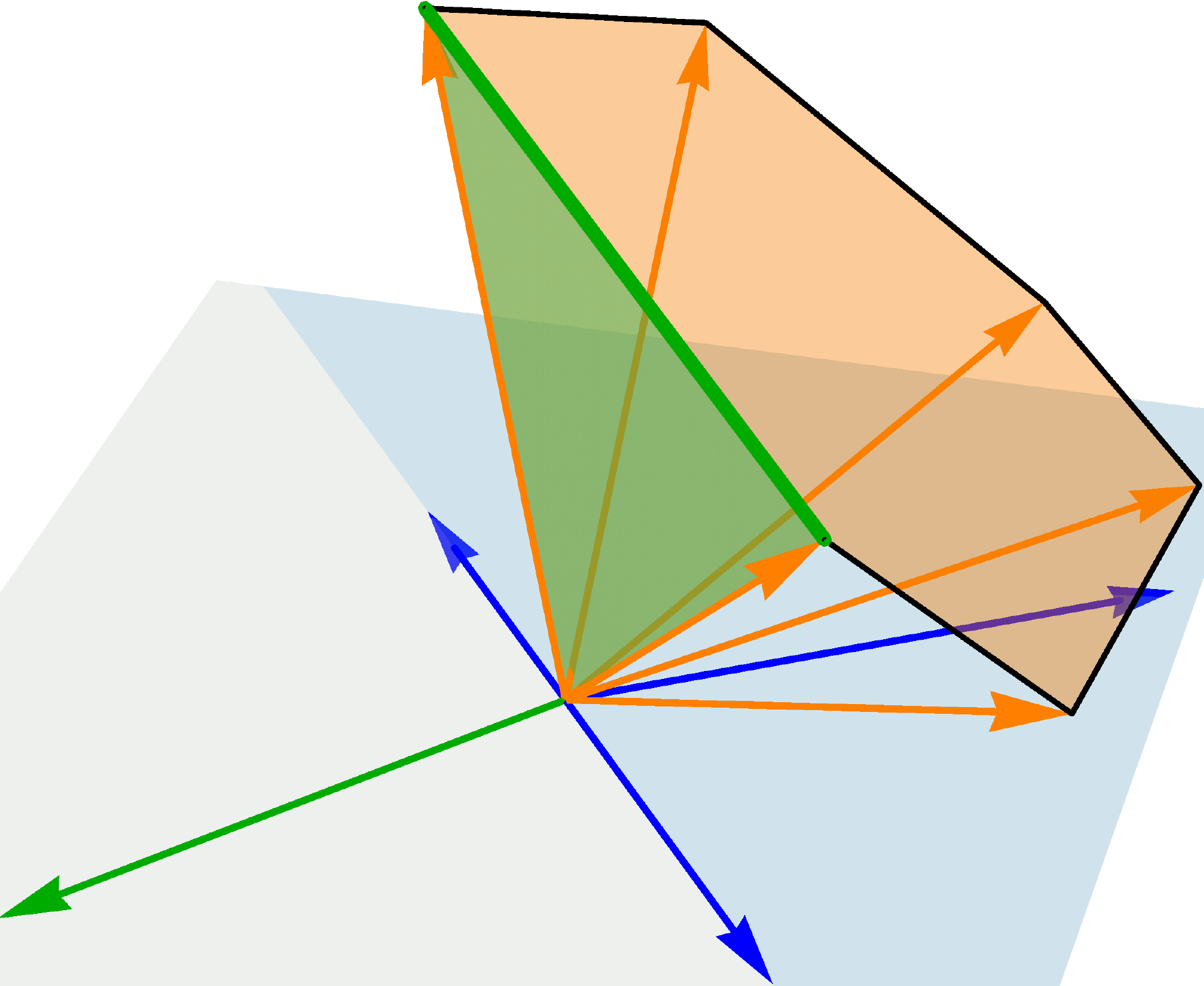}}
  
 \caption{Visualization of the CPT to find facets of a polytope
where the normal vector obeys some constraints. 
(a) We aim to find the facets of the two-dimensional polytope $P$ 
where the normal vector $\vec{b}_P$ has a fixed scalar product with 
some vector $\vec{g}$. The facet and the normal vector that fulfill 
this constraint in the given example are shown in green. 
(b) We embed the polytope $P$ in a plane in three-dimensional space, 
where the plane does not contain the origin $\vec 0$. The vertices of $P$ define 
rays (orange arrows) that define the cone $C$. The polytope $P$ is then 
the intersection of the cone $C$ with the plane, and each of its facets 
relates to a facet of $C$ (green) in a unique manner. 
(c) The initial constraint on the facet of $P$ can be translated to conditions
on the facets of the cone. A facet of $P$ fulfills the constraint if and only 
if the corresponding facet of $C$ has a normal vector $\vec{b}_C$ that obeys 
a linear constraint $G \vec{b}_C=0$, where $G$ is some matrix. Geometrically, 
this means that $\vec{b}_C$ has to lie in a certain plane (light-grey). Then, 
we project the rays of $C$ (orange) into that plane (blue arrows) to define 
a cone $\tilde C$ in the low-dimensional plane by taking projected rays as 
generators. By construction, facets of $C$ which obey the constraint
are also facets of the projected cone $\tilde C$.
(d) Finally, we find the facets of $\tilde C$ and check which ones 
correspond to facets of $C$. From the facets of $C$ that meet the 
conditions we can then compute the corresponding facets of $P$. 
In our example, $\tilde C$ is a half plane (light-blue) and has 
only one facet with the normal vector in green. It is also the 
normal vector of a facet of $C$ (green). Note that in the given 
example $\tilde C$ is already generated by three rays and the other three 
rays are redundant.
}
  \label{fig1}
\end{figure}

We consider a $D$-dimensional polytope $P$ and affine 
conditions on the facet normal vectors. Each of these
conditions can be written in the form
\begin{align}
\vec g^T_k \vec b_P = \gamma_k,
\label{eq-startingcond}
\end{align}
that is, the normal vector $\vec{b}_P$  has a fixed scalar 
product with some vector $\vec{g}_k$, see Fig.~\ref{fig1}(a). 
The task is to find {\it all} the facet normal vectors that 
obey these constraints.

In the first step, we construct a cone $C$ in $D+1$-dimensional 
space that maintains a one-to-one correspondence to the polytope. 
One way to achieve this is to prepend one fixed coordinate to 
every vertex $\vec v_i$,
\begin{align}
\vec v_i \mapsto \vec w_i = \binom{1}{\vec v_i}.
\end{align}
and to define the cone $C$ as the conic hull of the rays 
$\vec w_i$. In this way, the polytope can be seen as the 
intersection of the cone $C$ with the hyperplane defined 
by $x_0 = 1$. This relationship is illustrated in Fig.~\ref{fig1}(b). 
Notably, there is a one-to-one correspondence between 
the facets of the cone and those of the polytope, and 
a normal vector $\vec{b}_P$ of a polytope facet translates 
to a normal vector $\vec{b}_C$ of a cone facet. This is 
easy to formulate in the H-representation of the polytope 
and the corresponding cone. Let
\begin{align}
\vec x^T \vec b_P \le -\beta
\end{align}
be a facet-defining inequality of the polytope. Then, 
\begin{align}
(1 \; \vec x^T) \binom{\beta}{\vec b_P} \le 0
\end{align}
is the corresponding facet-defining inequality of the 
cone. Note that then $\binom{\beta}{\vec b_P}$ is the normal
vector of the facet, and any normal vector of a facet can be written
in this way. 

The correspondence between the polytope and the cone 
enables us to work with the cone instead of the polytope. 
This construction also allows to write the conditions in
Eq.~(\ref{eq-startingcond}) in a linear form, namely
\begin{align}
(g_{0,k} \; \vec g^T_k) \binom{\beta}{\vec b_P} = 0,
\end{align}
where we set $g_{0,k} = -\frac{\gamma_k}{\beta}$. 

Collecting all the facet conditions yields the 
linear matrix equation
\begin{align}
G \vec{b}_C=0, \label{eq-conecond}
\end{align}
where the $k$-th row of $G$ is the row vector $(g_{0,k} \; \vec g^T_k)$ 
and the facet normal vector of the cone is $\vec{b}_C = \binom{\beta}{\vec b_P}$.
Geometrically speaking, Eq.~(\ref{eq-conecond}) defines a hyperplane 
through the origin in which the facet normal vectors of the cone must lie 
in order to comply with the conditions in Eq.~(\ref{eq-startingcond}), see also
Fig.~\ref{fig1}(c).

The key observation is that in this situation we can define a new cone
$\tilde C$, such that if $\vec{b}_C$ is a facet normal vector of $C$ 
that obeys the constraints, then $\vec{b}_C$ is also a facet normal 
vector of $\tilde C$. This is done by projecting the rays of $C$ down 
to the subspace of vectors obeying $G \vec{v}=0$, see Fig.~\ref{fig1}(c) 
and Fig.~\ref{fig3ext}.

\begin{figure}[t]
\includegraphics[width=.3\textwidth]{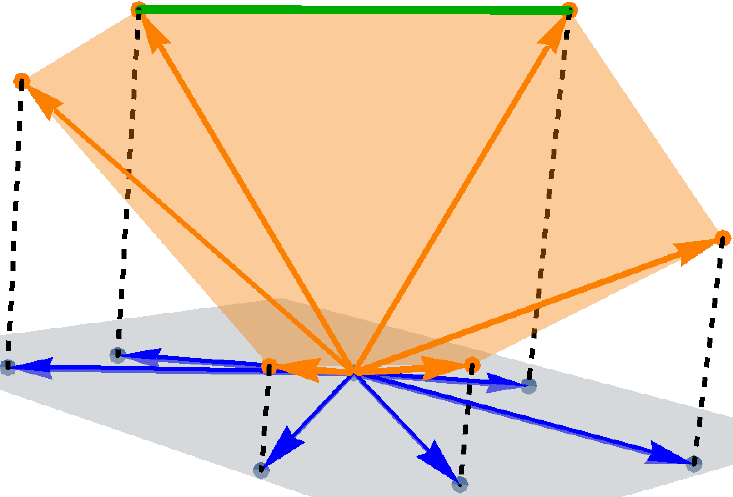}
\caption{Detailed view of the of the projection of the rays of 
the cone $C$ onto the plane where $G \vec{b}_C=0$. The projected 
rays (in blue) generate a new cone $\tilde C$. Depending on the 
conditions, this new cone may have a significantly reduced dimension
compared with $C$.}
\label{fig3ext}
\end{figure}

This subspace is the kernel of $G$ and it is spanned by a set of
$K$ vectors $\vec t_i$, where $K$ is the dimension of the kernel. 
The advantage of $\tilde C$ is that its dimension $K$ is typically 
considerably smaller than the dimension of $C$. Additionally, $\tilde C$ 
has typically much less rays than $C$. That makes it easier to find 
all the facets of $\tilde C$, compared with $C$, see Fig.~\ref{fig1}(d).

In practice, the first step in the construction of $\tilde C$ is to 
define the $(D+1)\times K$ matrix $T$, whose $K$ columns of length 
$D+1$ are given by the vectors $\vec t_j$, so we have $T_{ij} =
[\vec t_j]_i$. With this, we define the rays of $\tilde C$ as 
\begin{align}
\vec{\tilde w_i}^T = \vec{w_i}^T T.
\end{align}
Note that this is a projection if the vectors $\vec t_j$ 
form an orthonormal basis of the kernel of $G$. However, this is 
in general not necessary and in practice it can be
preferable to pick vectors $\vec t_j$ with integer 
coefficients, so that the rays of $\tilde C$
have integer coefficients, if the rays of $C$ have 
integer coefficients. In this way, one does not need
to worry about the precision of the numerical 
calculations.

Similarly to the rays of $C$, we can express its 
facet normal vectors $\vec b_C$ (that satisfy the
constraints) in the basis $\vec t_j$ of the kernel of 
$G$ as
\begin{align}
\vec b_C = T\vec b_{\tilde C}, \label{eq:btilde}
\end{align}
where $\vec b_{\tilde C}$ is a vector of dimension $K$.

The following theorem is the central result and 
establishes the previously claimed relation between 
$\tilde C$ and $C$. Namely, it states that any time 
$\vec b_C$ is a facet normal vector of $C$ that 
satisfies the conditions, $\vec b_{\tilde C}$ is 
a facet normal vector of $\tilde C$. In this way, the 
facet normal vectors of $\tilde C$ are the only relevant
vectors that one needs to consider.

\begin{theorem}
Let $C = \conv(\{\vec w_i\})$ be a cone and  $\vec b_C$ 
a facet normal vector of $C$ that satisfies $G \vec b_C = 0$ 
for some matrix $G$. With $T$ and $\vec b_{\tilde C}$ defined 
as above, we define the cone $\tilde C = \conv(\{\vec{\tilde w_i}\})$ 
of dimension $K$ with $\vec{\tilde w_i}^T = \vec{w_i}^T T$. 
Then $\vec b_{\tilde C}$ defines a facet of $\tilde C$.
\end{theorem}
\begin{proof}

We prove the statement in three steps.
(1) The inequality $\vec{\tilde w_i}^T \vec b_{\tilde C} \le 0$ holds, 
since \eq{btilde} together with the definition
of the $\vec{\tilde w_i}$ implies 
\begin{align}
\vec{\tilde w_i}^T \vec b_{\tilde C} = \vec{w_i}^T \vec b_C \label{eq:tildenotilde}
\end{align} and $\vec{w_i}^T \vec b \le 0$ because $\vec b_C$ is facet defining.

(2) The vector $\vec b_{\tilde C}$ defines a face of $\tilde C$, as one can directly see from \eq{tildenotilde}.
With $K = \dim(\ker G)$, the dimension of the face is at most $K-1$, since it is
contained in the $K-1$ dimensional subspace $\{\vec x\,| \,\vec x^T \vec b_{\tilde C} = 0\}$.

(3) The vector $\vec b_{\tilde C}$ defines a facet of $\tilde C$. That is, the dimension of the face is exactly $K-1$. 

Let $B$ be the $M \times(D+1)$ matrix that contains all M rays $\vec w_i$ as rows that
 fulfill $\vec w_i^T \vec b_C = 0$. Since $\vec b_C$ is a facet normal vector,
$B$ has rank $D$. Accordingly, $BT$ is the $M \times K$ matrix that contains all rays $\vec{\tilde w_i}^T$ as 
rows  that fulfill $\vec{\tilde w_i}^T \vec b_{\tilde C} = 0$.
Showing that $\vec{\tilde b}$ defines a facet is equivalent to showing that $\rank(BT) = K-1$. We now prove the latter by contradiction.
Assume there exist two linearly independent vectors
$\vec{\tilde b}, \vec{\tilde c}$ that satisfy $BT\vec{\tilde b} =
BT\vec{\tilde c} = 0$. Thus, $T\vec{\tilde b}$ and $T\vec{\tilde c}$ lie
in the kernel of $B$. Since $\rank(B) = D$, the kernel is one-dimensional, so
we can write $T \vec{\tilde c} = \ell T \vec{\tilde b}$ for some real number
$\ell$. This implies $T(\vec{\tilde c} + \ell \vec{\tilde b}) = 0$. Because
$\vec{\tilde b}$ and $\vec{\tilde c}$ are linearly independent, the kernel
of $T$ has at least dimension one, which is impossible because $T$ has full
column rank.
\end{proof}

The facets of interest of the polytope $P$ can now be found by finding 
the facets of $\tilde C$ first, calculating potential facets of $C$ via 
\eq{btilde}, transforming these into potential facets of $P$ and finally 
checking which of the found inequalities define  facets of $P$. Note that
it is computationally simple to check whether a given candidate is a facet, 
one just needs to compute the dimension of the surface. 

\subsection{Some facts about the local polytope}
After discussing polytopes in general we should also address 
the class of polytopes of interest in this paper, namely local 
polytopes. The aim of a Bell test is to discriminate between 
classical and non-classical behaviors. As the name suggests, 
the behavior of a system in a Bell test captures the relevant 
data from the Bell test experiment. A behavior is a vector
consisting of so-called correlations. In general, a correlation 
is the conditional probability of some combination of local 
measurement outcomes given a combination of local measurements 
that are performed in a time-like separated
manner. This includes the case that some parties may not perform 
any measurement on their subsystem, in which case the measured 
correlation is called a marginal correlation, in contrast to a 
full-body correlation, which captures the case in which every party performs a non-trivial measurement. 

Note that there is no need to consider the conditional probabilities for all possible combinations of outcomes given a combination of local measurement settings. The reason for this is that there are two rules that hold for the behaviors. The first rule is that the probabilities for the different combinations of local measurement outcomes given the same local measurement settings will sum up to one. The second rule is the no-signaling principle. It states that the probability for any local measurement outcome does not depend on the choices for
the measurement settings on the other subsystems. This reduces the dimension of  a behavior to 
\begin{align}
D(n, \vec m, o) = \prod_{i=1}^n \left(\sum_{j=1}^{m_i}(o_{ij} - 1) +1\right) - 1,
\label{eq:dim}
\end{align} 
where $n$ denotes the number of parties, $m_i$ is the number of 
measurement settings of party $i$ and $o_{ij}$ is the number of 
possible measurement outcomes for setting $j$ of party $i$ \cite{pironio2005}. The numbers $n, m_i, o_{ij}$ define a so-called
scenario for a Bell test. All scenarios considered in this paper 
feature dichotomic measurements exclusively, so $o_{ij} = 2, \; \forall i,j$. 

In this case it is sufficient to consider the expectation value
of the product of a combination of local measurement outcomes, 
given the measurement settings as a correlation instead of a conditional probability. One readily verifies that
this leads to the correct number of independent correlations according to
\eq{dim}.

Classical behaviors behave in accordance with a local hidden variable (LHV)
model. In a LHV model, the local measurement outcomes are fundamentally
predetermined by a local hidden variable $\lambda$. However, if the local hidden variable takes on different values through different runs of an experiment, this will cause the observed behavior to be a mixture of the local deterministic
behaviors. Mathematically speaking, the classical behaviors lie in a polytope,
the local polytope, which is defined as the convex hull of its vertices, the
local deterministic behaviors. Therefore, in order to define the local polytope
for a specific scenario, one must enumerate all local deterministic behaviors,
where each local deterministic behavior corresponds to one combination of local 
deterministic assignments of a measurement outcome to each measurement setting,
yielding 
\begin{align}
N = \prod_{i,j} o_{ij}
\end{align}
vertices in total.

\section{Results}
In this section, we present detailed results 
on generalizations of existing Bell inequalities 
to more parties using the cone-projection technique. 
We start with general considerations, then we 
present details on the found Bell inequalities 
for four relevant scenarios.

\subsection{General considerations}
For simplicity, we only consider Bell inequalities for 
scenarios in which every measurement has two possible 
outcomes. This allows us to write the Bell inequalities 
in a notation using observables. To simplify the notation 
of Bell inequalities that include marginal terms, that is, 
correlations that do not involve all parties, we define the
zeroth setting of every party to refer to a measurement 
that always yields the result $1$. Hence, $\langle A_1 \rangle$ 
is equivalent to $\langle A_1 B_0 C_0 ... \rangle$. Moreover, 
this allows us to write any three-party Bell inequality in the 
form $\sum_{ijk} b_{ijk} \langle A_i B_j C_k \rangle \ge 0$, 
where $b_{000}$ accounts for the constant term. 

As we already explained in the introduction, we call a 
facet-defining $(n+1)$-party Bell inequality that is 
reducible to an $n$-party Bell inequality a generalization 
of the latter. In mathematical terms, we can formulate this 
condition for a facet-defining bipartite Bell inequality  
that is to be generalized to a three-party Bell inequality 
as follows. The three-party Bell inequality 
$\sum_{i,j,k} b_{ijk} \langle A_i B_j C_k\rangle \ge 0$ 
is called a generalization of the Bell inequality 
$\sum_{ij} \beta_{ij} \langle A_i B_j \rangle \ge 0$ if it is
facet-defining and there exist deterministic choices $\xi_k
\in \{-1,1\}$ for the outcomes of Charlie such that
\begin{align}
\sum_{i,j,k} b_{ijk} \xi_k \langle A_i B_j \rangle = \sum_{ij} \beta_{ij}
\langle A_i B_j \rangle .
\end{align} 
This condition implies the following: If
$\gamma$ with $\gamma_{ij} = \langle A_i B_j \rangle$ is a 
behavior that saturates the two party Bell inequality, that 
is $\sum_{ij} \beta_{ij} \gamma_{ij} = 0$, then the behavior 
$g_{ijk} = \gamma_{ij}\xi_k$ will saturate its generalization. 
Since $g$ is a three-party extension of the behavior $\gamma$, 
we call it an extended behavior. 

Note that the condition that the extended behavior should saturate the generalized Bell inequality is a linear constraint and can thus be 
handled with the CPT. However, there can be 
generalized Bell inequalities for every possible deterministic 
choice of assignments of outcomes $\xi_k$ to Charlies measurement
settings, so we have to find all the generalized Bell inequalities 
for every of these choices. Fixing one choice, we obtain one extended behavior for every saturating behavior of the inequality we seek to generalize. 

Additionally, we require that generalized Bell inequalities should obey symmetries that are characteristic for the inequality they generalize. 
This narrows down the search and increases the similarity between an inequality and its generalizations.

In the following, we present generalizations of several inequalities 
that will be discussed in their respective sections. For every 
inequality we found, we conduct the same numerical analysis. We 
compute the quantum mechanical violation for qubits, for qutrits
and for the second and (if possible) third level of the NPA-hierarchy
\cite{navascues2007}. For the latter, we used the ncpol2sdpa package 
by Peter Wittek \cite{wittek2015}. For the qubit and qutrit violations, 
we provide lower bounds using a seesaw algorithm that alternates
between optimizing the measurement settings of one of the parties and optimizing the state. Every of these optimization steps is a semidefinite program. Details can be found in Appendix~A.

For meaningful comparisons between the inequalities, we compute four
quantities for every Bell inequality $\langle \mathcal B \rangle \ \equiv \sum_{ijk\neq 000} -b_{ijk} \langle A_i B_j C_k \rangle \le b_{000} $ that 
will guide the following discussion. 

The first quantity is the {\it relative qutrit violation}  by 
which the maximal value of the Bell expression $\mathcal B$ achievable 
with qutrits exceeds the maximal classical value. It is defined as
\begin{align}
m_Q& = \frac{\max_{\text{qutrit}} \expect{\mathcal B}}{b_{000}} - 1 \label{eq-m1}.
\end{align}
We are interested in this quantity as a signature to identify 
inequalities whose quantum violation is particularly strong. 
In this regard, it would be more informative to consider the 
quantum violation with a higher-dimensional system. However, 
the computations become more demanding which is why we restrict 
our quantum mechanical analysis to qubit and qutrit systems.

The second quantity we consider is the {\it algebraic-classical ratio}
that is defined analogously and quantifies, by how much the 
algebraic maximum of the Bell expression exceeds the classical bound 
in relation to the classical bound. It is calculated as
\begin{align}
m_A &= \frac{\sum_{i,j,k} |b_{ijk}| - b_{000}}{b_{000}}.
\end{align}

As a third interesting quantity, we compute the {\it NPA-qutrit 
ratio} that quantifies the relative margin between the 
maximal value of the Bell expression in the third level of the 
NPA hierarchy and the maximal value achievable with qutrits. For 
some of the Bell inequalities we found, it was unfeasible to
compute the third level of the NPA hierarchy. In this situation 
we resorted to the value obtained by the second level of the NPA 
hierarchy. Every time this is the case it is stated explicitely. 
The NPA-qutrit ratio is calculated as
\begin{align}
m_N &= \frac{\max_{\text{NPA}} \expect{\mathcal B}}{\max_{\text{qutrit}} \expect{\mathcal B}} - 1.
\end{align}

Lastly, we consider the {\it qutrit-qubit ratio} $m_{32}$ by which 
the maximal value of
the Bell expression $\mathcal B$ achieved by qutrits exceeds the corresponding
value for qubits. Mathematically, it is defined as
\begin{align}
m_{32} &= \frac{\max_{\text{qutrit}} \expect{\mathcal B}}{\max_{\text{qubit}}
\expect{\mathcal B}} - 1.\label{eq-m4}
\end{align}
It is natural to state these relative margins in percent, in which case $m_{32}
= 10\%$ means that the maximal value of the Bell expression is 
$10\%$ larger for qutrit systems than for qubit systems.

In the following, we report those inequalities that exhibit the
biggest value for one of the four relative quantities. A big 
qutrit-qubit ratio $m_{32}$ makes the inequality potentially 
interesting for experimental discrimination between qubit and 
qutrit states (i.e. device-independent dimension witnessing), 
while a large NPA-qutrit ratio $m_N$ suggests that the maximal 
quantum mechanical violation of the inequality may not be achievable 
with qutrit states. A large relative qutrit violation $m_Q$ hints at 
a particularly strong violation of the inequality in quantum mechanics,  while a large algebraic-classical ratio $m_A$ shows that observers that play by the rules of classical physics are especially limited in 
obtaining a large expectation value for the Bell expression in 
comparison with hypothetical observers who do not suffer any
physical limitations.

Finally, while we defined the quantities in this subsection for
three parties, all the definitions can easily be extended to
four-party Bell inequalities.

\subsection{Investigation of I3322 generalizations}
The I3322 inequality is the next simplest Bell inequality to consider
after the standard CHSH inequality and it is the only facet-defining 
Bell inequality in the bipartite scenario with three dichotomic measurements per party (besides lifted versions of the CHSH inequality). 
It reads
\begin{align}
\mean{A_1} & - \mean{A_2} + \mean{B_1} - \mean{B_2} - 
\mean{(A_1 - A_2)(B_1 - B_2)}
\nonumber
\\ 
& + \mean{(A_1 + A_2)B_3} + \mean{A_3(B_1 + B_2)}
\le 4.
\label{eq-i3322}
\end{align}
The I3322 inequality, or I3322 for short, was first discovered by
Froissart in 1981 \cite{froissart1981} and more than twenty years 
later rediscovered by Sliwa \cite{sliwa2003} and, independently, 
Collins and Gisin \cite{collins2004}. It has sparked interest 
because its maximal quantum violation is not achievable with 
qubits. Curiously, there is no difference in the maximal quantum 
violations achieved by qubits and qutrits. In fact, the local 
dimensions of Alice's and Bob's must be at least twelve in order 
to observe a larger quantum violation than the one achievable 
with qubits \cite{pal2010}. In a previous work 
we already provided a list of 3050 classes of Bell 
inequalities that generalize the I3322 inequality
\cite{bernards2020}. In this paper we supply a list of 
all the inequalities including our numerical results 
in the Supplementary Material (see Appendix \ref{app-supp} for details). For simplicity, we adopt 
a notation that takes into account that the inequalities
are symmetric under arbitrary permutations of parties. We 
write $(ijk) := \langle A_i B_j C_k \rangle + \text{permutations 
that yield different terms}$. For example 
$(011) = \langle A_1 B_1 \rangle + \langle B_1 C_1 \rangle + \langle A_1 C_1\rangle $. 

Let us discuss these inqualities. To start with a rough overview, 
the classical bounds 
$b_{000}$ of the Bell inequalities range from 8 to 321. 
Fig.~\ref{fig-i3322bars} depicts the classical bound for 
every Bell inequality and how much this bound is exceeded 
by the quantum states we found as well as the upper bound 
thereof provided by the NPA hierarchy. We find that all of 
the inequalities can be violated for qubit states. The 
inequalities are ordered by simplicity, measured as
length of their expression in the above stated notation.

\begin{figure}[t]
\includegraphics[width=.53\textwidth]{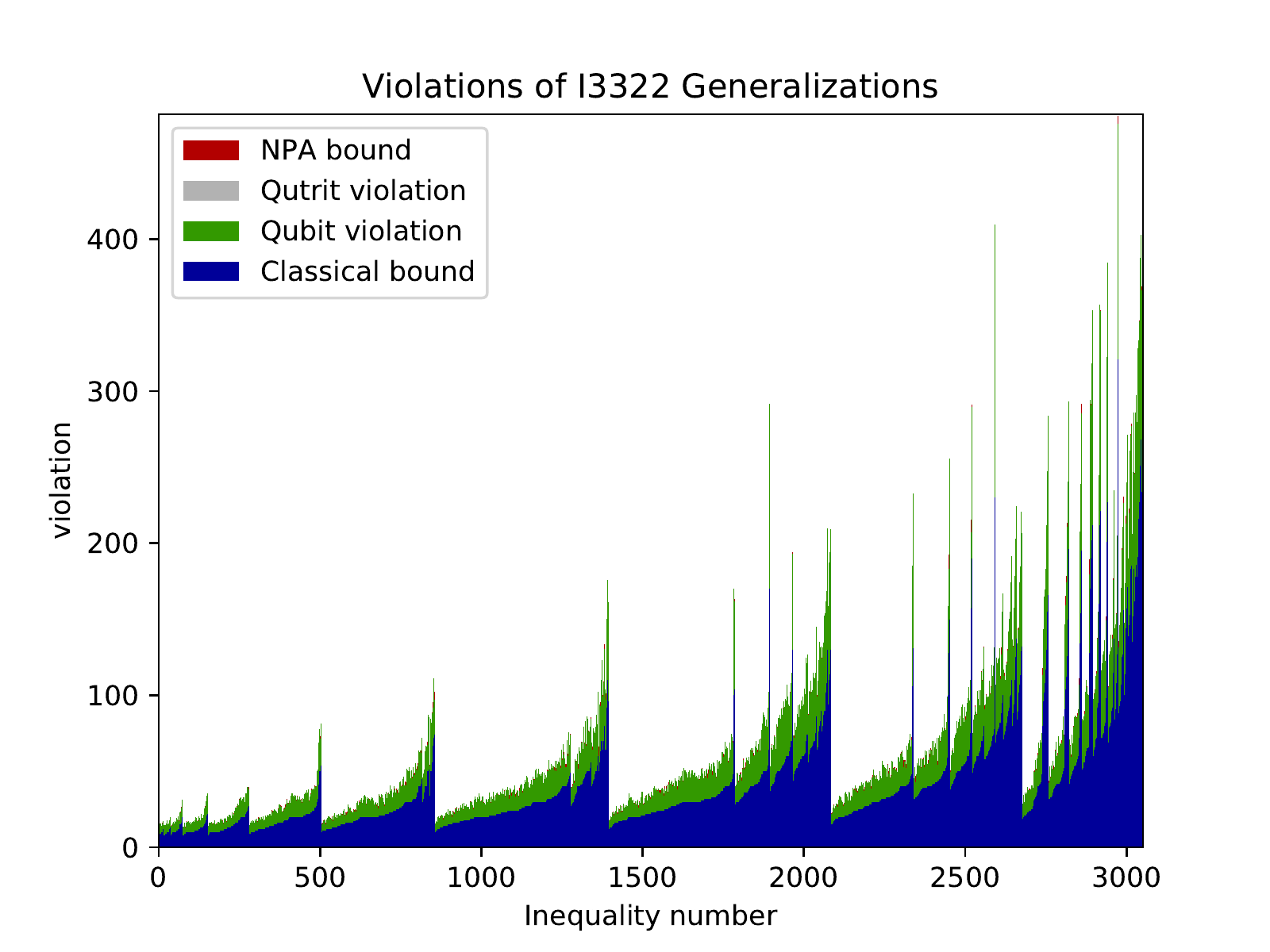}
\caption{The classical bounds of the inequalities 
$\sum_{ijk\neq 000} -b_{ijk} \langle A_i B_j C_k \rangle \le b_{000}$
 and their violation by qubits,
qutrits and the NPA hierarchy of third level. The Bell inequalities are sorted by the number of terms they incorporate. The orange and blue bars that represent the qutrit violations and NPA bounds can barely be
seen because they are mostly hidden behind the bars representing 
the qubit violations.}
\label{fig-i3322bars}
\end{figure}

Now let us consider some specific Bell inequalities. The largest 
qutrit-classical ratio $m_Q$ is achieved for Bell inequality 
number 1 as listed in the Supplementary Material, see Appendix \ref{app-supp}, 
which reads
\begin{align}
 & (110) + (210) - (211) - (220)  \notag \\
&- (222) +2 (331) +2 (332) \le 8.
\label{i3322-1}
\end{align}
For this inequality, the maximal value attainable with qubits coincides with the
upper bound of 16 given by the
NPA-hierarchy. It is worth noting that this inequality reduces to Mermin's
inequality if the first two settings of each party are chosen equal. The
inequality then is maximally violated for the Greenberger-Horne-Zeilinger state
$|GHZ\rangle = (|000\rangle + |111\rangle)/\sqrt{2}$ and the optimal measurement settings
for Mermin's inequality. In fact, this has been already discussed
in Ref.~\cite{bernards2020} and more details on states and settings 
can be found there. 

We observe the biggest qutrit-qubit ratio $m_{32}$ for Bell
inequality number 400, which reads
\begin{align}
& +5 (100) + (110) -5 (111) +3 (200) -3 (210) +3 (211)\notag \\
 & -2 (220) +2 (221) -2 (300) + (310) + (311)\notag \\
 & +2 (322) - (330)\notag \\
 & - (333)\le 18
\label{i3322-400}
\end{align}
For qubits, QM can lead to a violation of $20.928$, whereas 
for qutrits a maximal violation of $21.157$ can be obtained, 
which is less than the value $21.238$ obtained using the third
level of the NPA-hierarchy.

We find the biggest NPA-qutrit ratio $m_N$ for Bell inequality 
number 1507, which reads
\begin{align}
&8 (100) -4 (110) +3 (111) +4 (200) -3 (210) \notag\\
&+2 (211) - (220) +2 (221) -2 (222) - (300) - (310) \notag\\
&+3 (311) +2 (320) - (321) + (322) - (331) - (333) \le 21 .
\label{i3322-1507}
\end{align}
The third level of the NPA hierarchy yields the value $24.079$, 
which is a significant improvement over the value $26.299$ 
provided by the second level of the NPA hierarchy. However, 
with qutrits we could only achieve a violation of $23.603$. 
With qubits, the violation was even lower at $23.249$.

The largest algebraic-classical ratio $m_A$ between the classical bound $12$ and the algebraic bound $86$ is found for 
Bell inequality number 532. It reads
\begin{align}
&3 (100) -(110) - (111) -4 (200) +2 (210) -(220) \notag\\
&+ (221) + (222) -3 (300) +2 (310) - (320) + (321)\notag\\ 
&- (322) + (331) + (332) \le 12 
\label{i3322-532}
\end{align}
All four ratios are listed for the inequalities above in Tab.~\ref{tab-i3322}.

\begin{table}[t]
\tabcolsep=0ex
\rowcolors{2}{white}{lightgray}
\begin{tabularx}{.9\columnwidth}{XX!{\tabsep}lXXXX}
Eq. & Number &$\;$& $m_Q$ & $m_{32}$ & $m_N$ & $m_A$ \\\toprule
(\ref{i3322-1}) & 1 && 100.0 & 0.0 & -0.0 & 250.0 \\
(\ref{i3322-400}) &400 && 17.54 & 1.09 & 0.38 & 433.33 \\
(\ref{i3322-1507}) & 1507 && 12.4 & 0.61 & 2.01 & 514.29 \\
(\ref{i3322-532}) & 532 && 18.69 & 0.0 & 0.35 & 616.67 \\\bottomrule
\end{tabularx}

\caption{This table shows the relative qutrit violation $m_Q$, 
the qutrit-qubit ratio $m_{32}$, the NPA-qutrit ratio $m_N$, 
and the algebraic-classical ratio $m_A$ as well as those 
generalizations of the I3322 inequality for which one of 
these margins is the largest. The values of the margins 
is stated in percent. For example, for inequality number 1, 
the maximal value of the Bell expression is 16 for qutrit 
systems, exceeding the classical bound of 8 by 8, yielding $m_Q
= 100\%$.}
\label{tab-i3322}
\end{table}

\subsection{Three-party generalizations of the I4422 inequality}
When Collins and Gisin discovered the I3322 inequality, they noticed that
this Bell inequality is a member of a family of Bell inequalities which they called I$mm$22, where $m \geq 2$. This family contains the CHSH inequality, or I2222, as a special case. Inequalities of this family 
with more measurement settings generalize inequalities with fewer measurement settings in the following sense: If one substitutes 
$A_1 = 1$ and $B_3 = 1$ in the I3322 inequality, the inequality
effectively reduces to a CHSH inequality. Essentially the same 
procedure also works for the I4422 inequality \cite{collins2004}. 
In the notation introduced above, it reads
\begin{align}
& 2\, (01) + 2\, (02) + (03) - (11) - (12)\notag\\ -& (13) - (14) - (22) - (23)
+ (24) + (33) \le 7.
\label{eq-i4422}
\end{align}
If we substitute $A_3 = B_3 = 1$, followed by $A_2 \rightarrow -A_2$, $B_2
\rightarrow -B_2$ and $A_4 \rightarrow -A_3$, $B_4 \rightarrow -B_3$, we arrive exactly at the I3322 inequality as stated in Eq.~(\ref{eq-i3322}). 

In contrast to the I3322 inequality we observe a difference between the maximal qubit-violation and the maximal qutrit-violation. For qubits, the largest value we could achieve for the Bell expression is $8$, whereas for qutrits it is $8.15$, which matches the upper bound obtained using the third level of the NPA hierarchy up to numerical precision. Obtaining a value of $8$ is compatible with the above mentioned contraints that will reduce the inequality to the I3322 inequality. This means that the additional measurement setting of I4422 seems to lead to no advantage 
for achieving a large violation of the inequality, as long as only qubits are concerned. Perhaps surprisingly, the generalizations of
I4422 we are going to present do not seem to inherit the feature 
of being able to discriminate between qubit and qutrit systems by 
means of their maximal violations.

Before presenting some generalizations of I4422, we should clarify 
that we did not tackle this problem in full generality. When looking 
for generalizations of the I3322 inequality to more parties, the only symmetry we demanded was invariance under permutations of parties. 
Finding generalizations of the I4422 inequality is, however, computationally more involved, which is why we demand a second 
symmetry in order to make the problem tractable. When choosing a 
symmetry, one needs to be careful to not impose too strong constraints. 
For example, a three-party inequality that is invariant under party permutations can never be the generalization of a two party inequality that is not symmetric with respect to the parties. We therefore have to take the symmetries of the inequality we seek to generalize, in our case I4422, into consideration.

Besides being invariant under exchange of the parties $A$ and $B$, 
I4422 has a
second symmetry, namely, if we swap Alice's first and second setting while
simultaneously relabeling the outcomes of Bob's fourth measurement, then this
leaves the inequality invariant. For convenience, we write the symmetries in the
following symbolic way: 
\begin{enumerate}
\item
$ A \leftrightarrow B$
\item $ A_1 \leftrightarrow A_2,
\quad B_4 \rightarrow -B_4$. 
\end{enumerate}

For three-party generalizations of I4422, the following two symmetries seem
natural. The first one is
\begin{enumerate}
\item $ A \leftrightarrow B $
\item $ A \leftrightarrow C $
\item $ A_1 \leftrightarrow A_2$ , $ B_4 \rightarrow -B_4 $, $ C_4 \rightarrow
-C_4$.
\end{enumerate}
However, there exist no non-trivial facet inequalities that generalize I4422 and 
at the same time meet the above conditions.

We therefore instead look for I4422 generalizations with the following symmetry:
\begin{enumerate}
\item $ A \leftrightarrow B $
\item $ A \leftrightarrow C $
\item $ A_1 \leftrightarrow A_2$ , $ B_1 \leftrightarrow B_2 $, $ C_4 \rightarrow
-C_4$.
\end{enumerate}
The first two symmetries together establish that the generalizations we are
going to find are symmetric under arbitrary permutations of the three parties.
The third symmetry resembles the second symmetry of I4422.
In total, there are 13 classes of Bell inequalities that generalize I4422 and
also exhibit the above symmetries. They are presented in Appendix B.

\begin{figure}[t]
\includegraphics[width=.5\textwidth]{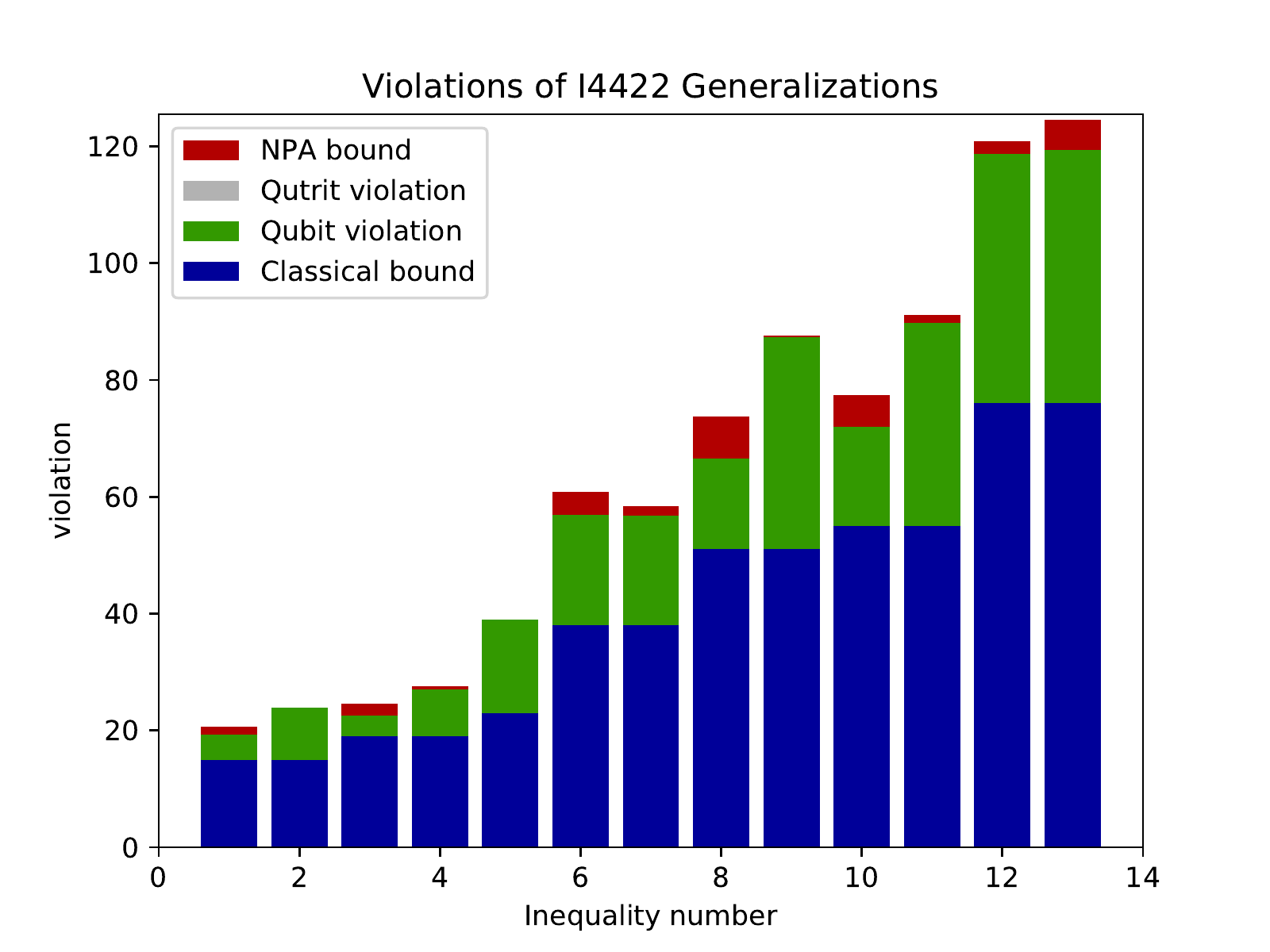}
\caption{The classical bounds of the I4422 generalizations
 and their violation by qubits,
qutrits and the NPA hierarchy (second level).}
\label{fig-i4422bars}
\end{figure}

Fig.~\ref{fig-i4422bars} shows
the classical bounds and the violations we found for the 13 inequalities.
The more complex scenario of tripartite
generalizations of I4422 also comes with the drawback that the third level of the NPA
hierarchy is already hard to compute. We therefore only compute the second level of the NPA hierarchy.
Tab.~\ref{tab-i4422} shows the relative margins as defined in Eqs.~(\ref{eq-m1} -
\ref{eq-m4}). For all of the inequalities the maximal qubit and qutrit
violations we found were the same. However, we cannot be certain that we managed
to find the maximal qutrit violation, except from the case of inequality number
5, where the quantum violations attained the upper bound given by the NPA
hierarchy up to numerical precision.

\begin{table}[t]
\tabcolsep=0ex
\rowcolors{2}{white}{lightgray}
\begin{tabularx}{.9\columnwidth}{XX!{\tabsep}lXXXX}
Eq. & Number &$\;$& $m_Q$ & $m_{32}$ & $m_N$ & $m_A$ \\\toprule
(\ref{eq-i4422-1}) & 1 & &28.87 & 0.0 & 6.81 & 360.0 \\ 
(\ref{eq-i4422-2}) & 2 & &59.05 & 0.0 & 0.24 & 413.33 \\
(\ref{eq-i4422-3}) & 3 & &18.39 & 0.0 & 9.49 & 315.79 \\
(\ref{eq-i4422-4}) & 4 & &42.11 & 0.0 & 2.18 & 357.89 \\
(\ref{eq-i4422-5}) & 5 & &69.57 & 0.0 & 0.0 & 434.78 \\
(\ref{eq-i4422-6}) & 6 & &49.83 & 0.0 & 6.82 & 400.0 \\
(\ref{eq-i4422-7}) & 7 & &49.56 & 0.0 & 2.75 & 378.95 \\
(\ref{eq-i4422-8}) & 8 & &30.39 & 0.0 & 10.85 & 341.18 \\
(\ref{eq-i4422-9}) & 9 & &71.33 & 0.0 & 0.33 & 403.92 \\
(\ref{eq-i4422-10})& 10& & 30.87 & 0.0 & 7.56 & 327.27 \\
(\ref{eq-i4422-11})& 11& & 63.34 & 0.0 & 1.39 & 385.45 \\
(\ref{eq-i4422-12})& 12& & 56.21 & 0.0 & 1.76 & 373.68 \\
(\ref{eq-i4422-13})& 13& & 57.12 & 0.0 & 4.26 & 405.26 \\\bottomrule
\end{tabularx}
\caption{This table shows the relative qutrit violation $m_Q$, the qutrit-qubit ratio $m_{32}$, the NPA-qutrit ratio  $m_N$, and the algebraic-classical ratio $m_A$ for every generalization of I4422 that satisfies the symmetry described in the text. The values of the ratios are expressed 
in percent. Numbers refers to the number under which the corresponding inequality is listed in the Appendix B.}
\label{tab-i4422}
\end{table}

\subsection{Three-party hybrid CHSH-I3322 generalizations}

To show that also hybrid scenarios can be studied, we 
find all 476 
 three-party facet Bell inequalities with the following
properties: Alice and Bob have three measurement settings
each and Charlie has
only two. Also, there
are deterministic assignments to the outcomes of Alice's measurements, such that
the inequality effectively reduces to the CHSH inequality
\begin{align}
\mean{A_1 B_1} +
\mean{A_1 B_2} +
\mean{A_2 B_1} -
\mean{A_2 B_2} 
\le 2.
\label{eq-chsh2}
\end{align}
Further, there are deterministic assignments to the outcomes of Charlie's
measurements, such that the inequality reduces to the I3322 inequality. Lastly,
we require the inequality to be symmetric under exchange of Alice and Bob. These
conditions lead to 476 inequalities, these are given
in the Supplementary Material, see Appendix \ref{app-supp}.

As one can see in Fig.~\ref{fig-hybridbars} all of the Bell inequalities are violated in quantum
mechanics. 
The strongest violation by qutrit states is achieved for inequality $47$. It reads
\begin{align}
 & \langle C_1\rangle+ \langle C_2\rangle+ \langle B_1\rangle\notag \\ 
&- \langle B_1C_2\rangle- \langle B_2\rangle+ \langle B_2C_1\rangle+ \langle A_1\rangle\notag \\ 
&- \langle A_1C_2\rangle+2 \langle A_1B_1C_2\rangle+ \langle A_1B_2\rangle- \langle A_1B_2C_1\rangle\notag \\ 
&- \langle A_1B_3C_1\rangle- \langle A_1B_3C_2\rangle- \langle A_2\rangle+ \langle A_2C_1\rangle\notag \\ 
&+ \langle A_2B_1\rangle- \langle A_2B_1C_1\rangle- \langle A_2B_2\rangle+2 \langle A_2B_2C_1\rangle\notag \\ 
&- \langle A_2B_2C_2\rangle+ \langle A_2B_3\rangle- \langle A_2B_3C_2\rangle- \langle A_3B_1C_1\rangle\notag \\ 
&- \langle A_3B_1C_2\rangle+ \langle A_3B_2\rangle- \langle A_3B_2C_2\rangle- \langle A_3B_3\rangle\notag \\ 
&- \langle A_3B_3C_1\rangle \le 6 
\label{hybrid-47}
\end{align}

\begin{figure}[t]
\includegraphics[width=.5\textwidth]{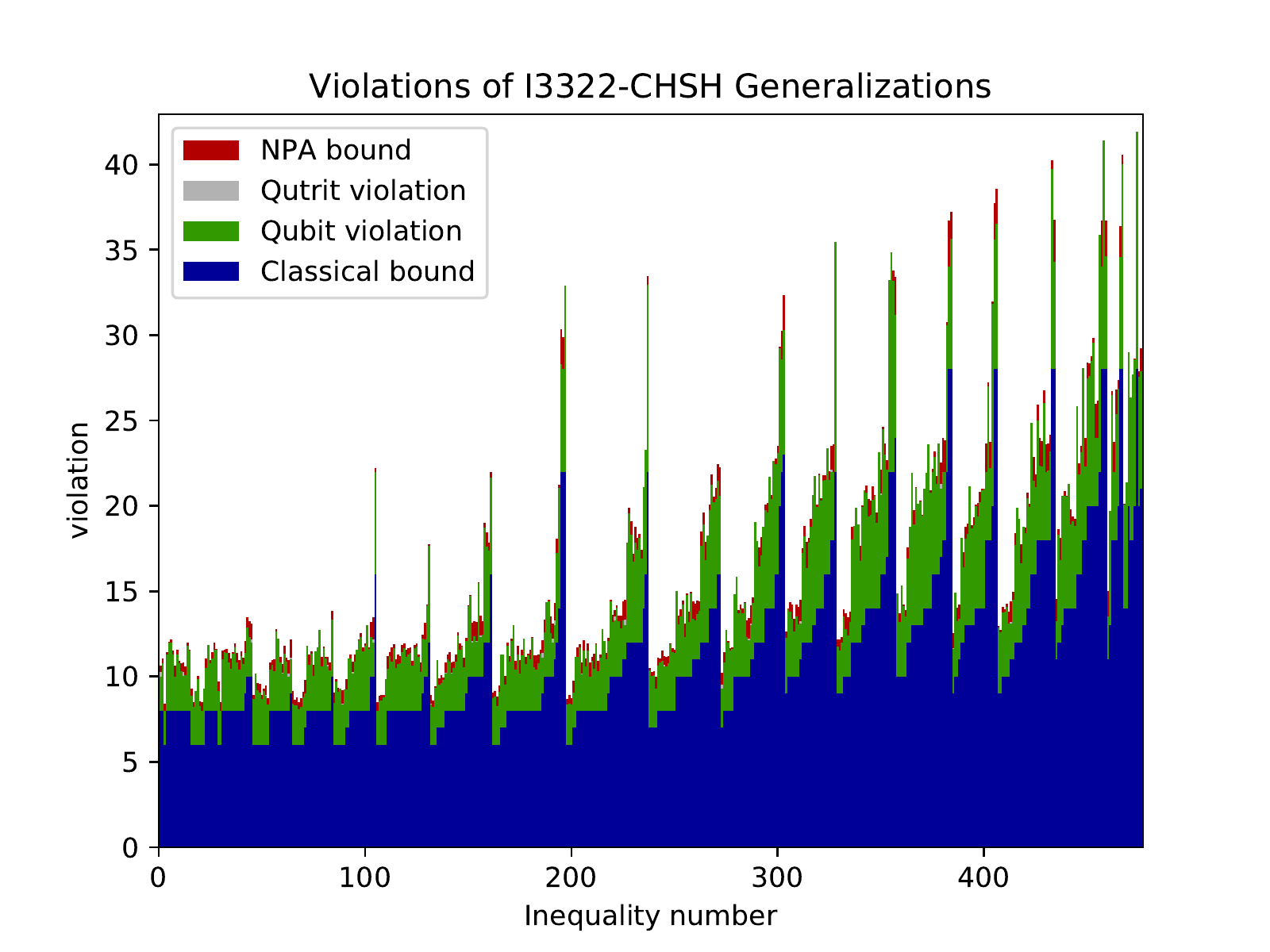}
\caption{The classical bounds of the inequalities 
$\sum_{ijk\neq 000} -b_{ijk} \langle A_i B_j C_k \rangle \le b_{000}$
and their violation by qubits,
qutrits and the NPA hierarchy of third level. See text for further details.}
\label{fig-hybridbars}
\end{figure}

We find the largest NPA-qutrit ratio $m_N$ for Bell inequality
number 314. It reads
\begin{align}
 &5\langle C_1\rangle +\langle C_2\rangle +2\langle B_1\rangle -2\langle
B_1 C_1\rangle\notag \\
&-2\langle B_2\rangle +2\langle B_2 C_1\rangle +\langle B_3 C_1\rangle 
+\langle B_3 C_2\rangle\notag\\
& +2\langle A_1\rangle -2\langle A_1 C_1\rangle +2\langle
A_1B_1C_1\rangle -2\langle A_1B_1C_2\rangle\notag\\
& +\langle A_1 B_2\rangle  -2\langle A_1B_2C_1\rangle +\langle
A_1B_2C_2\rangle\notag\\
& +\langle A_1 B_3\rangle
-2\langle A_1B_3C_1\rangle +\langle A_1B_3C_2 \rangle -2\langle A_2\rangle \notag\\&+2\langle A_2 C_1\rangle
+\langle A_2 B_1\rangle 
-2\langle A_2B_1C_1\rangle +\langle A_2B_1C_2\rangle\notag\\& -2\langle A_2 B_2\rangle 
+3\langle A_2B_2C_1\rangle +\langle A_2B_2C_2\rangle +\langle A_2 B_3\rangle \notag\\
&-3\langle A_2B_3C_1\rangle +\langle A_3 C_1\rangle 
+\langle A_3 C_2\rangle +\langle A_3 B_1\rangle\notag \\ &-2\langle A_3B_1C_1\rangle 
+\langle A_3B_1C_2\rangle +\langle A_3 B_2\rangle\notag \\& 
-3\langle A_3B_2C_1\rangle -2\langle A_3B_3C_1\rangle -2\langle A_3B_3C_2\rangle \le 12 . 
\label{hybrid-314}
\end{align}
We achieve the same violation for qubit and qutrit systems at $16.339$. The
third level of the NPA hierarchy yields $16.488$, improving the upper bound of $17.870$ from the second level NPA hierarchy significantly. Because the qubit and qutrit violations are the same, one may conjecture that the upper bound for the quantum mechanical violation provided by the NPA hierarchy is not tight for the third level.

We observe the biggest gap between qubits and qutrits for Bell inequality number 1. It reads
\begin{align}
&2 \langle B_1\rangle-2 \langle B_2\rangle+2 \langle A_1\rangle- \langle A_1B_1\rangle\notag \\
 &+ \langle A_1B_1C_1\rangle+ \langle A_1B_2\rangle- \langle A_1B_2C_1\rangle+2 \langle A_1B_3C_2\rangle\notag \\
 &-2 \langle A_2\rangle+ \langle A_2B_1\rangle- \langle A_2B_1C_1\rangle- \langle A_2B_2\rangle\notag \\
 &+ \langle A_2B_2C_1\rangle+2 \langle A_2B_3C_2\rangle+2 \langle A_3B_1C_2\rangle+2 \langle A_3B_2C_2\rangle\notag \\
 &\le 8
\label{hybrid-1}
\end{align}
For this inequality we find a 
value of $10.000$ for qubits, whereas the violation of $10.286$ that we find for qutrits
coincides with the value of the third level of the NPA hierarchy up to numerical
precision.

Lastly, the biggest gap between the classical and the algebraic bound occurs for inequality $198$, which reads
\begin{align}
 &3 \langle C_1\rangle+ \langle C_2\rangle+2 \langle B_1\rangle\notag \\ 
&- \langle B_1C_1\rangle- \langle B_1C_2\rangle- \langle B_2\rangle+2 \langle B_2C_1\rangle\notag \\ 
&+ \langle B_2C_2\rangle+ \langle B_3\rangle- \langle B_3C_2\rangle+2 \langle A_1\rangle\notag \\ 
&- \langle A_1C_1\rangle- \langle A_1C_2\rangle- \langle A_1B_1\rangle+ \langle A_1B_1C_1\rangle\notag \\ 
&+ \langle A_1B_2\rangle- \langle A_1B_2C_1\rangle- \langle A_1B_3C_1\rangle+ \langle A_1B_3C_2\rangle\notag \\ 
&- \langle A_2\rangle+2 \langle A_2C_1\rangle+ \langle A_2C_2\rangle+ \langle A_2B_1\rangle\notag \\ 
&- \langle A_2B_1C_1\rangle- \langle A_2B_2\rangle- \langle A_2B_2C_2\rangle+ \langle A_2B_3\rangle\notag \\ 
&- \langle A_2B_3C_1\rangle+ \langle A_3\rangle- \langle A_3C_2\rangle- \langle A_3B_1C_1\rangle\notag \\ 
&+ \langle A_3B_1C_2\rangle+ \langle A_3B_2\rangle- \langle A_3B_2C_1\rangle \le
6.
\label{hybrid-198}
\end{align}
Compared to the classical bound of $6$, the algebraic bound is almost seven
times as large at $40$.
The numerical results are summarized in Tab.~\ref{tab-hybrid}.

\begin{table}[t]
\tabcolsep=0ex
\rowcolors{2}{white}{lightgray}
\begin{tabularx}{.9\columnwidth}{XX!{\tabsep}lXXXX}
Eq. & Number &$\;$& $m_Q$ & $m_{32}$ & $m_N$ & $m_A$ \\\toprule
(\ref{hybrid-47}) & 47 && 69.23 & 0.0 & 0.0 & 400.0 \\
(\ref{hybrid-1}) & 1 && 28.57 & 2.86 & 0.0 & 200.0 \\
(\ref{hybrid-314}) & 314 && 36.16 & 0.0 & 0.91 & 433.33 \\
(\ref{hybrid-198}) & 198 && 39.71 & 0.0 & 0.0 & 566.67 \\\bottomrule
\end{tabularx}
\label{tab-hybrid}
\caption{This table shows the relative qutrit violation $m_Q$, the qutrit-qubit ratio $m_{32}$, the NPA-qutrit ratio $m_N$, and the algebraic-classical ratio
$m_A$ as well as those hybrid generalizations of the I3322 and the CHSH inequality for which one of
these margins is the largest. The values of the ratios is stated in percent.}
\end{table}

\subsection{Four-party generalizations of a Guess-Your-Neighbors-Input inequality}

In this section we present generalizations of a guess your neighbors input
(GYNI) Bell inequality. GYNI inequalities give a classical upper bound for the winning probability of the nonlocal game that goes by the same name. The game is played with $N$ parties that are arranged in a ring and collaboratively play against the instructor. The instructor will supply every participant $i$ with an input bit $x_i$ and each participant then has to guess the input of his left neighbor. The game is won if the output $a_i$ of every player output matches the input
$x_{i+1}$ of their left neighbor. In order to achieve this goal, the players may
agree on a strategy before the game and they are also provided with the
probability distribution $q(\vec x)$ according to which the instructor chooses
the inputs. In the quantum version of the game, the parties additionally possess
a part of a shared quantum state. During the game, no communication is permitted. A GYNI inequality is 
an inequality of the form
\begin{align}
\sum_{\vec x} q(\vec x) P( a_i = x_{i+1} \forall i | \vec x) \le \omega_c,
\end{align}
whereby the left-hand-side expresses the average winning probability, which is
bounded by the average winning probability of the best classical strategy
$\omega_c$. Notably, Almeida et al \cite{almeida2010} showed that GYNI Bell
inequalities are never violated by quantum theory and discovered that a facet defining
Bell-inequality previously found by \'Sliwa is a GYNI inequality. This inequality,
written in observable notation, reads
\begin{align}
  &\mean{A_1 B_1}  +\mean{A_2 B_1} +\mean{A_1 B_2 }+\mean{A_2 B_2}  \nonumber
\\&+\mean{A_1 C_1} -\mean{A_2 C_1} +\mean{B_1 C_1} +\mean{A_1 B_1 C_1}  \nonumber
\\&-\mean{B_2 C_1} -\mean{A_2 B_2 C_1}
+\mean{A_1 C_2} -\mean{A_2 C_2} -\mean{B_1 C_2}  \nonumber 
\\&+\mean{A_2 B_1 C_2} +\mean{B_2 C_2}
-\mean{A_1 B_2 C_2} \le 4  \label{eq-gyni}
\end{align}
and is listed as inequality number $10$ in the paper by \'Sliwa \cite{sliwa2003}. Here
we present its generalizations. The inequality has the following symmetries:
\begin{enumerate}
\item $A_1 \leftrightarrow A_2$, $B_1 \leftrightarrow B_2$, $C_1 \rightarrow
-C_1$, $C_2 \rightarrow -C_2$
\item $A_1 \leftrightarrow A_2$, $ C_1 \leftrightarrow C_2 $, $A_i \rightarrow
-A_i \forall i$, $B_i \rightarrow -B_i \forall i$.
\end{enumerate}

\begin{figure}[t]
\includegraphics[width=.5\textwidth]{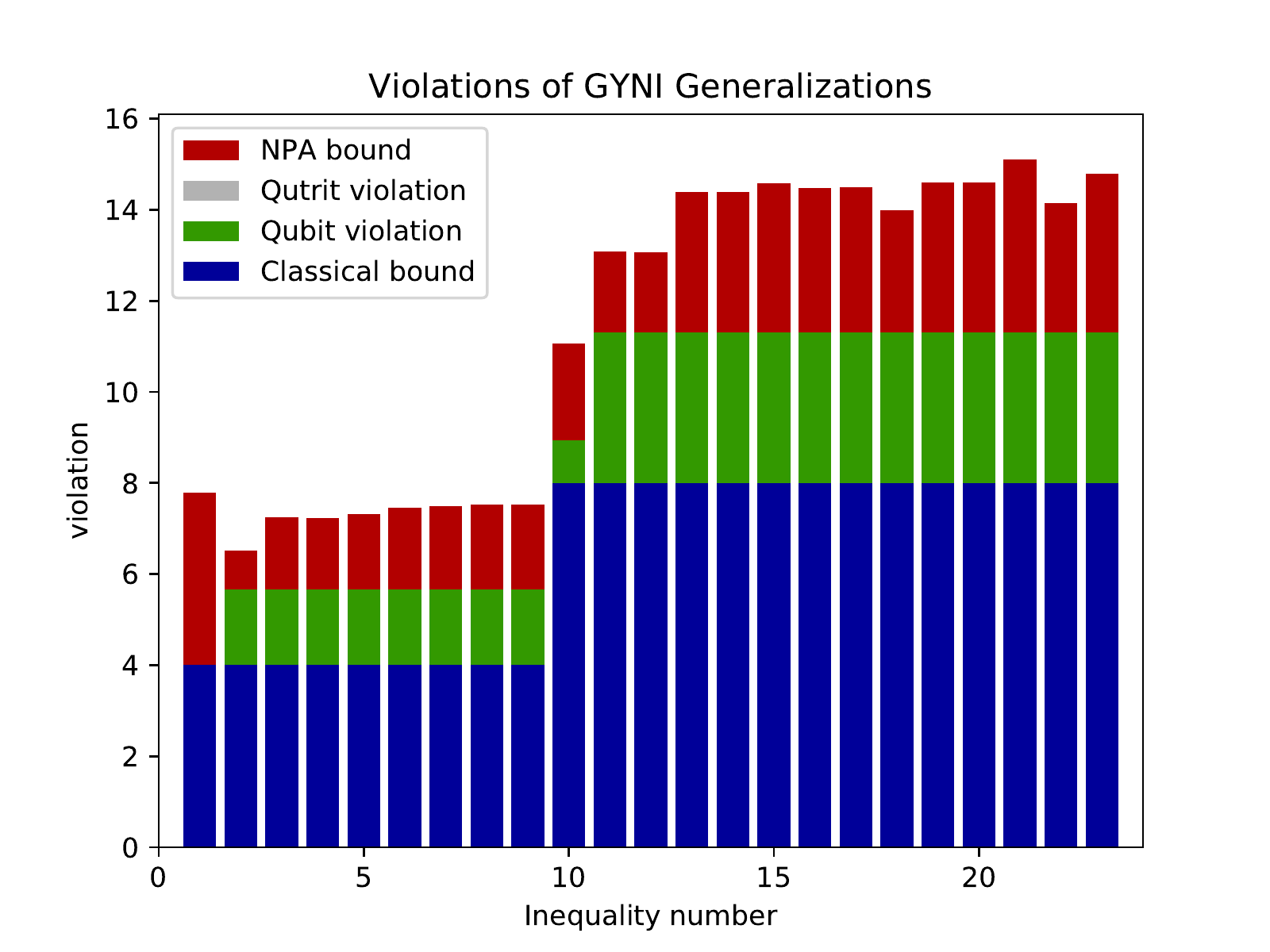}
\caption{The classical bounds of the inequalities 
$\sum_{ijkl\neq 0000} -b_{ijkl} \langle A_i B_j C_k D_l \rangle \le b_{0000}$
and their violation by qubits,
qutrits and the NPA hierarchy of third level.}
\label{fig-gynibars}
\end{figure}

We demand exactly the same symmetries for the generalization of the above
inequality. We find $23$ four-partite inequalities, the detailed expressions are
given in the Supplementary Material, see Appendix \ref{app-supp}. 
The first one is not a real four-partite Bell inequality, since it only includes a single measurement setting on the fourth party. In fact, it is of the form $ B (1 + D_1) \le 4(1 + D_1) $, where $B$ is the
left-hand-side of Eq.~(\ref{eq-gyni}) after $C_1$ and $C_2$ have been swapped and all outcomes of Bob have been relabeled. All the other inequalities are violated in quantum theory. Hence, none of the generalizations shares the characteristic feature of the  GYNI inequality.
Fig.~\ref{fig-gynibars} shows a plot with all classical bounds and their
violations for all 23 inequalities.

\section{Conclusion}
We discussed the cone-projection technique as a method to 
find facet Bell inequalities with normal vectors that 
obey linear constraints. We further introduced the concept 
of a generalization of a Bell inequality to more parties. 
The property of being a generalization of a specific Bell 
inequality can be formulated as a linear constraint on its 
normal vector. Using the cone-projection technique, we were able to find 3050
classes of generalizations of the I3322 inequality, 476 
classes of Bell inequalities that simultaneously generalize 
the CHSH inequality as well as the I3322 inequality, 13 
classes of Bell inequalities that generalize the I4422
inequality and 23 classes of Bell inequalities that 
generalize the Guess-Your-Neighbors-Input inequality 
first found in \cite{sliwa2003}. For all inequalities, 
we applied an extensive numerical analysis, providing
upper bounds to their quantum violations both for qubit 
and qutrit systems as well as upper bounds using the NPA 
hierarchy.

For future research, there are several open problems. First, 
it would be desirable to study some of the Bell inequalities
presented here in more detail. For instance, the question arises
whether they can be connected to some distributed information 
processing task. In addition, one may study the viability of an
experimental test of these inequalities. Second, it is very interesting
to study the cone-projection technique also for other scenarios, such
as  Bell inequalities that detect genuine multi-partite nonlocality or 
contextuality inequalities.

\acknowledgments
We thank Miguel Navascués, Chau Nguyen and Denis Rosset for
fruitful discussions. This work was supported by the
Deutsche Forschungsgemeinschaft (DFG, German Research 
Foundation, project numbers 447948357 and
440958198), the Sino-German Center for Research Promotion (Project M-0294), and the ERC (Consolidator Grant 683107/TempoQ).
FB acknowledges support from the House of Young Talents of the University
of Siegen.

\appendix

\section{\label{app-viol}
Details on finding lower bounds for qubit and qutrit
violations}
We provide lower bounds for the violations of the Bell inequalities achievable with qubits and qutrits, respectively. These lower bounds are violations achieved by specific states and measurements. 

Finding good lower bounds therefore involves two optimizations: One optimization over the possible states and another optimization over the measurement settings. The objective -- that is the expectation value of the Bell operator for the given state -- is however not
linear, if we optimize over all measurements at once. We therefore break the
optimization over the measurement settings down further to optimizing over them
party by party, thereby rendering the objective linear. We represent each
measurement by a positive operator-valued measure (POVM). Such a optimization
problem that features an objective that is linear in its arguments which are
positive-semidefinite operators is called a semidefinite program (SDP) and can
be solved efficiently. Likewise, the optimization over the
state-space while keeping the measurement settings constant is also an SDP.

This prompts a seesaw-algorithm: After initializing the measurement settings to
(pseudo-)random values, we alternate between updating the quantum state and
updating the measurement settings one party at a time while keeping the quantum state
constant. In this way, the value of the objective is guaranteed to increase
monotonically and will eventually converge, if the objective is bounded.
However, it oftentimes only converges to a local minimum. This is a problem we
cannot avoid, but we can aim for a rather good local optimum by running the
seesaw-algorithm many times, starting at a different initial state every time.
To save computational resources, we do not wait until the objective values reach
high precision. Instead, we stop after a few iterations, select the most
promising instance of the optimization and only continue the seesaw-algorithm
for this selected instance until the improvements of the objective value fall below some threshold.

\section{I4422 generalizations \label{app-i4422gens}}
We find 13 classes of Bell inequalities that are symmetric under party
permutations, generalize I4422 and further
exhibit the symmetry
$A_1 \leftrightarrow A_2, B_1 \leftrightarrow B_2, C_4 \rightarrow -C_4$.
In the notation for symmetric Bell inequalities that we introduced in the main
text, they read
\begin{widetext}
\begin{align}
  - (100) -2 (110) - (200) -2 (210) -2 (220) + (300) - (310)&\notag\\  \qquad\quad 
- (320) + (330) + (331) + (332) -3 (333) +2 (441) -2 (442) &\le 15  \qquad\quad\label{eq-i4422-1}\\
  - (100) -2 (110) - (111) - (200) -2 (210) + (211) -2 (220) - (221) + (222) &\notag\\  \qquad\quad+ (300) - (310) - (320) + (330) + (331) + (332) -3 (333) +2 (441) -2
(442) &\le 15  \qquad\quad\label{eq-i4422-2} \\
  -2 (100) -2 (110) -2 (200) -2 (210) -2 (220) - (300) -2 (310) -2 (320) &\notag\\  \qquad\quad
+(330) - (333) +2 (441) -2 (442) &\le 19  \qquad\quad\label{eq-i4422-3} \\
  -2 (100) -2 (110) - (111) -2 (200) -2 (210) + (211) -2 (220) - (221)  &\notag\\  \qquad\quad
+(222) - (300) -2 (310) -2 (320) + (330) - (333) +2 (441) -2 (442) &\le 19  \qquad\quad\label{eq-i4422-4} \\
  +3 (100) -3 (110) +2 (111) +3 (200) -3 (210) -3 (220) +2 (221) + (300) - (310)&\notag\\  \qquad\quad + (311) - (320) + (321) + (322) + (330) - (331) - (332) + (333) -4
(441) +4 (442) &\le 23  \qquad\quad\label{eq-i4422-5} 
\end{align}
\begin{align}
  +3 (100) + (110) -3 (111) +3 (200) + (210) + (211) + (220) -3 (221) + (222) +4 (300)&\notag\\  \qquad\quad
 +3 (310) -3 (311) +3 (320) -3 (321) -3 (322) -6 (330) - (331) -
(332) +12 (333) +4 (441) -4 (442) &\le 38  \qquad\quad\label{eq-i4422-6} \\
  +3 (100) + (110) - (111) +3 (200) + (210) - (211) + (220) - (221) -
(222) +4 (300) +3 (310) -3 (311)&\notag\\  \qquad\quad +3 (320) -3 (321) -3 (322) -6 (330) - (331) -
(332) +12 (333) +4 (441) -4 (442) &\le 38  \qquad\quad\label{eq-i4422-7} \\
  +2 (100) -5 (110) +2 (200) -5 (210) -5 (220) + (300) -4 (310) +3 (311) -4 (320)&\notag\\  \qquad\quad +3 (321) +3 (322) + (330) -2 (331) -2 (332) -3 (333) -8 (441) +8 (442) &\le 51  \qquad\quad\label{eq-i4422-8} \\
  +2 (100) -5 (110) -4 (111) +2 (200) -5 (210) +4 (211) -5 (220) -4 (221) +4 (222)+ (300)&\notag\\  \qquad\quad  -4 (310) +3 (311) -4 (320) +3 (321) +3 (322) + (330) -2 (331) -2
(332) -3 (333) +8 (441) -8 (442) &\le 51  \qquad\quad\label{eq-i4422-9}
\end{align}
\begin{align}
  + (100) -5 (110) + (200) -5 (210) -5 (220) - (300) -5 (310) +3 (311) -5
(320)&\notag\\  \qquad\quad +3 (321) +3 (322) + (330) -3 (331) -3 (332) - (333) -8
(441) +8 (442) &\le 55  \qquad\quad\label{eq-i4422-10} \\
  + (100) -5 (110) -4 (111) + (200) -5 (210) +4 (211) -5 (220) -4 (221) +4 (222) - (300)&\notag\\  \qquad\quad -5 (310) +3 (311) -5 (320) +3 (321) +3 (322) + (330) -3 (331) -3
(332) - (333) +8 (441) -8 (442) &\le 55  \qquad\quad\label{eq-i4422-11} \\
  +7 (100) + (110) - (111) +7 (200) + (210) - (211) + (220) - (221) -
(222) +6 (300) +7 (310) -7 (311)&\notag\\  \qquad\quad +7 (320) -7 (321) -7 (322) -12 (330) - (331)
- (332) +22 (333) -8 (441) +8 (442) &\le 76  \qquad\quad\label{eq-i4422-12} \\
  +7 (100) + (110) -5 (111) +7 (200) + (210) +3 (211) + (220) -5 (221) 
+3 (222) +6 (300) +7 (310) -7 (311)&\notag\\  \qquad\quad +7 (320) -7 (321) -7 (322) -12 (330) - (331)
- (332) +22 (333) +8 (441) -8 (442) &\le 76.  \qquad\quad\label{eq-i4422-13}
\end{align}
\end{widetext}

\section{Content of the online material \label{app-supp}}
The further online material (included in the source files of this arxiv
submission) contains the generalizations of the I3322 inequality, the I4422
inequality, the Guess-Your-Neighbors-Input inequality (\eq{gyni}), and the
hybrid I3322-CHSH Bell inequalities. For convenience, the online material
includes four files, one for each class of Bell inequalities:
\begin{itemize}
 \item {\tt i3322gensquant.txt}  This file contains the generalizations of the
I3322 inequality.
  \item {\tt i4422gensquant.txt}  This file contains the generalizations of the
I4422 inequality.
  \item {\tt hybridgensquant.txt}  This file contains the
hybrid I3322-CHSH inequalities.
  \item {\tt gynigensquant.txt}  This file contains the generalizations of the
GYNI inequality in Eq.~(\ref{eq-gyni}).
\end{itemize}
Each file lists the Bell inequalities together with their respective algebraic
bounds as well as the bounds obtained from the second level of the NPA-hierarchy
and, with the exception of I4422 generalizations, also from the third level of
the NPA-hierarchy. Additionally, we provide lower bounds on the qubit violations and
qutrit violations as well as the corresponding qubit state and qubit settings
that lead to the listed qubit violation.

\bibliography{biblio}

\begin{thebibliography}{22}%
\makeatletter
\providecommand \@ifxundefined [1]{%
 \@ifx{#1\undefined}
}%
\providecommand \@ifnum [1]{%
 \ifnum #1\expandafter \@firstoftwo
 \else \expandafter \@secondoftwo
 \fi
}%
\providecommand \@ifx [1]{%
 \ifx #1\expandafter \@firstoftwo
 \else \expandafter \@secondoftwo
 \fi
}%
\providecommand \natexlab [1]{#1}%
\providecommand \enquote  [1]{``#1''}%
\providecommand \bibnamefont  [1]{#1}%
\providecommand \bibfnamefont [1]{#1}%
\providecommand \citenamefont [1]{#1}%
\providecommand \href@noop [0]{\@secondoftwo}%
\providecommand \href [0]{\begingroup \@sanitize@url \@href}%
\providecommand \@href[1]{\@@startlink{#1}\@@href}%
\providecommand \@@href[1]{\endgroup#1\@@endlink}%
\providecommand \@sanitize@url [0]{\catcode `\\12\catcode `\$12\catcode
  `\&12\catcode `\#12\catcode `\^12\catcode `\_12\catcode `\%12\relax}%
\providecommand \@@startlink[1]{}%
\providecommand \@@endlink[0]{}%
\providecommand \url  [0]{\begingroup\@sanitize@url \@url }%
\providecommand \@url [1]{\endgroup\@href {#1}{\urlprefix }}%
\providecommand \urlprefix  [0]{URL }%
\providecommand \Eprint [0]{\href }%
\providecommand \doibase [0]{http://dx.doi.org/}%
\providecommand \selectlanguage [0]{\@gobble}%
\providecommand \bibinfo  [0]{\@secondoftwo}%
\providecommand \bibfield  [0]{\@secondoftwo}%
\providecommand \translation [1]{[#1]}%
\providecommand \BibitemOpen [0]{}%
\providecommand \bibitemStop [0]{}%
\providecommand \bibitemNoStop [0]{.\EOS\space}%
\providecommand \EOS [0]{\spacefactor3000\relax}%
\providecommand \BibitemShut  [1]{\csname bibitem#1\endcsname}%
\let\auto@bib@innerbib\@empty
\bibitem [{\citenamefont {Brunner}\ \emph {et~al.}(2014)\citenamefont
  {Brunner}, \citenamefont {Cavalcanti}, \citenamefont {Pironio}, \citenamefont
  {Scarani},\ and\ \citenamefont {Wehner}}]{brunnerbellreview}%
  \BibitemOpen
  \bibfield  {author} {\bibinfo {author} {\bibfnamefont {N.}~\bibnamefont
  {Brunner}}, \bibinfo {author} {\bibfnamefont {D.}~\bibnamefont {Cavalcanti}},
  \bibinfo {author} {\bibfnamefont {S.}~\bibnamefont {Pironio}}, \bibinfo
  {author} {\bibfnamefont {V.}~\bibnamefont {Scarani}}, \ and\ \bibinfo
  {author} {\bibfnamefont {S.}~\bibnamefont {Wehner}},\ }\href {\doibase
  10.1103/RevModPhys.86.419} {\bibfield  {journal} {\bibinfo  {journal} {Rev.
  Mod. Phys.}\ }\textbf {\bibinfo {volume} {86}},\ \bibinfo {pages} {419}
  (\bibinfo {year} {2014})}\BibitemShut {NoStop}%
\bibitem [{\citenamefont {Scarani}(2019)}]{scarani2019}%
  \BibitemOpen
  \bibfield  {author} {\bibinfo {author} {\bibfnamefont {V.}~\bibnamefont
  {Scarani}},\ }\href@noop {} {\emph {\bibinfo {title} {{B}ell Nonlocality}}}\
  (\bibinfo  {publisher} {Oxford University Press},\ \bibinfo {year}
  {2019})\BibitemShut {NoStop}%
\bibitem [{\citenamefont {Shalm}\ \emph {et~al.}(2015)\citenamefont {Shalm},
  \citenamefont {Meyer-Scott}, \citenamefont {Christensen}, \citenamefont
  {Bierhorst}, \citenamefont {Wayne}, \citenamefont {Stevens}, \citenamefont
  {Gerrits}, \citenamefont {Glancy}, \citenamefont {Hamel}, \citenamefont
  {Allman}, \citenamefont {Coakley}, \citenamefont {Dyer}, \citenamefont
  {Hodge}, \citenamefont {Lita}, \citenamefont {Verma}, \citenamefont
  {Lambrocco}, \citenamefont {Tortorici}, \citenamefont {Migdall},
  \citenamefont {Zhang}, \citenamefont {Kumor}, \citenamefont {Farr},
  \citenamefont {Marsili}, \citenamefont {Shaw}, \citenamefont {Stern},
  \citenamefont {Abell\'an}, \citenamefont {Amaya}, \citenamefont {Pruneri},
  \citenamefont {Jennewein}, \citenamefont {Mitchell}, \citenamefont {Kwiat},
  \citenamefont {Bienfang}, \citenamefont {Mirin}, \citenamefont {Knill},\ and\
  \citenamefont {Nam}}]{shalm2015}%
  \BibitemOpen
  \bibfield  {author} {\bibinfo {author} {\bibfnamefont {L.~K.}\ \bibnamefont
  {Shalm}}, \bibinfo {author} {\bibfnamefont {E.}~\bibnamefont {Meyer-Scott}},
  \bibinfo {author} {\bibfnamefont {B.~G.}\ \bibnamefont {Christensen}},
  \bibinfo {author} {\bibfnamefont {P.}~\bibnamefont {Bierhorst}}, \bibinfo
  {author} {\bibfnamefont {M.~A.}\ \bibnamefont {Wayne}}, \bibinfo {author}
  {\bibfnamefont {M.~J.}\ \bibnamefont {Stevens}}, \bibinfo {author}
  {\bibfnamefont {T.}~\bibnamefont {Gerrits}}, \bibinfo {author} {\bibfnamefont
  {S.}~\bibnamefont {Glancy}}, \bibinfo {author} {\bibfnamefont {D.~R.}\
  \bibnamefont {Hamel}}, \bibinfo {author} {\bibfnamefont {M.~S.}\ \bibnamefont
  {Allman}}, \bibinfo {author} {\bibfnamefont {K.~J.}\ \bibnamefont {Coakley}},
  \bibinfo {author} {\bibfnamefont {S.~D.}\ \bibnamefont {Dyer}}, \bibinfo
  {author} {\bibfnamefont {C.}~\bibnamefont {Hodge}}, \bibinfo {author}
  {\bibfnamefont {A.~E.}\ \bibnamefont {Lita}}, \bibinfo {author}
  {\bibfnamefont {V.~B.}\ \bibnamefont {Verma}}, \bibinfo {author}
  {\bibfnamefont {C.}~\bibnamefont {Lambrocco}}, \bibinfo {author}
  {\bibfnamefont {E.}~\bibnamefont {Tortorici}}, \bibinfo {author}
  {\bibfnamefont {A.~L.}\ \bibnamefont {Migdall}}, \bibinfo {author}
  {\bibfnamefont {Y.}~\bibnamefont {Zhang}}, \bibinfo {author} {\bibfnamefont
  {D.~R.}\ \bibnamefont {Kumor}}, \bibinfo {author} {\bibfnamefont {W.~H.}\
  \bibnamefont {Farr}}, \bibinfo {author} {\bibfnamefont {F.}~\bibnamefont
  {Marsili}}, \bibinfo {author} {\bibfnamefont {M.~D.}\ \bibnamefont {Shaw}},
  \bibinfo {author} {\bibfnamefont {J.~A.}\ \bibnamefont {Stern}}, \bibinfo
  {author} {\bibfnamefont {C.}~\bibnamefont {Abell\'an}}, \bibinfo {author}
  {\bibfnamefont {W.}~\bibnamefont {Amaya}}, \bibinfo {author} {\bibfnamefont
  {V.}~\bibnamefont {Pruneri}}, \bibinfo {author} {\bibfnamefont
  {T.}~\bibnamefont {Jennewein}}, \bibinfo {author} {\bibfnamefont {M.~W.}\
  \bibnamefont {Mitchell}}, \bibinfo {author} {\bibfnamefont {P.~G.}\
  \bibnamefont {Kwiat}}, \bibinfo {author} {\bibfnamefont {J.~C.}\ \bibnamefont
  {Bienfang}}, \bibinfo {author} {\bibfnamefont {R.~P.}\ \bibnamefont {Mirin}},
  \bibinfo {author} {\bibfnamefont {E.}~\bibnamefont {Knill}}, \ and\ \bibinfo
  {author} {\bibfnamefont {S.~W.}\ \bibnamefont {Nam}},\ }\href {\doibase
  10.1103/PhysRevLett.115.250402} {\bibfield  {journal} {\bibinfo  {journal}
  {Phys. Rev. Lett.}\ }\textbf {\bibinfo {volume} {115}},\ \bibinfo {pages}
  {250402} (\bibinfo {year} {2015})}\BibitemShut {NoStop}%
\bibitem [{\citenamefont {Hensen}\ \emph {et~al.}(2015)\citenamefont {Hensen},
  \citenamefont {Bernien}, \citenamefont {Dréau}, \citenamefont {Reiserer},
  \citenamefont {Kalb}, \citenamefont {Blok}, \citenamefont {Ruitenberg},
  \citenamefont {Vermeulen}, \citenamefont {Schouten}, \citenamefont {Abellan},
  \citenamefont {Amaya}, \citenamefont {Pruneri}, \citenamefont {Mitchell},
  \citenamefont {Markham}, \citenamefont {Twitchen}, \citenamefont {Elkouss},
  \citenamefont {Wehner}, \citenamefont {Taminiau},\ and\ \citenamefont
  {Hanson}}]{hensen2015}%
  \BibitemOpen
  \bibfield  {author} {\bibinfo {author} {\bibfnamefont {B.}~\bibnamefont
  {Hensen}}, \bibinfo {author} {\bibfnamefont {H.}~\bibnamefont {Bernien}},
  \bibinfo {author} {\bibfnamefont {A.}~\bibnamefont {Dréau}}, \bibinfo
  {author} {\bibfnamefont {A.}~\bibnamefont {Reiserer}}, \bibinfo {author}
  {\bibfnamefont {N.}~\bibnamefont {Kalb}}, \bibinfo {author} {\bibfnamefont
  {M.}~\bibnamefont {Blok}}, \bibinfo {author} {\bibfnamefont {J.}~\bibnamefont
  {Ruitenberg}}, \bibinfo {author} {\bibfnamefont {R.}~\bibnamefont
  {Vermeulen}}, \bibinfo {author} {\bibfnamefont {R.}~\bibnamefont {Schouten}},
  \bibinfo {author} {\bibfnamefont {C.}~\bibnamefont {Abellan}}, \bibinfo
  {author} {\bibfnamefont {W.}~\bibnamefont {Amaya}}, \bibinfo {author}
  {\bibfnamefont {V.}~\bibnamefont {Pruneri}}, \bibinfo {author} {\bibfnamefont
  {M.}~\bibnamefont {Mitchell}}, \bibinfo {author} {\bibfnamefont
  {M.}~\bibnamefont {Markham}}, \bibinfo {author} {\bibfnamefont
  {D.}~\bibnamefont {Twitchen}}, \bibinfo {author} {\bibfnamefont
  {D.}~\bibnamefont {Elkouss}}, \bibinfo {author} {\bibfnamefont
  {S.}~\bibnamefont {Wehner}}, \bibinfo {author} {\bibfnamefont
  {T.}~\bibnamefont {Taminiau}}, \ and\ \bibinfo {author} {\bibfnamefont
  {R.}~\bibnamefont {Hanson}},\ }\href {\doibase 10.1038/nature15759}
  {\bibfield  {journal} {\bibinfo  {journal} {Nature}\ }\textbf {\bibinfo
  {volume} {526}},\ \bibinfo {pages} {682} (\bibinfo {year}
  {2015})}\BibitemShut {NoStop}%
\bibitem [{\citenamefont {Giustina}\ \emph {et~al.}(2015)\citenamefont
  {Giustina}, \citenamefont {Versteegh}, \citenamefont {Wengerowsky},
  \citenamefont {Handsteiner}, \citenamefont {Hochrainer}, \citenamefont
  {Phelan}, \citenamefont {Steinlechner}, \citenamefont {Kofler}, \citenamefont
  {Larsson}, \citenamefont {Abell\'an}, \citenamefont {Amaya}, \citenamefont
  {Pruneri}, \citenamefont {Mitchell}, \citenamefont {Beyer}, \citenamefont
  {Gerrits}, \citenamefont {Lita}, \citenamefont {Shalm}, \citenamefont {Nam},
  \citenamefont {Scheidl}, \citenamefont {Ursin}, \citenamefont {Wittmann},\
  and\ \citenamefont {Zeilinger}}]{giustina2015}%
  \BibitemOpen
  \bibfield  {author} {\bibinfo {author} {\bibfnamefont {M.}~\bibnamefont
  {Giustina}}, \bibinfo {author} {\bibfnamefont {M.~A.~M.}\ \bibnamefont
  {Versteegh}}, \bibinfo {author} {\bibfnamefont {S.}~\bibnamefont
  {Wengerowsky}}, \bibinfo {author} {\bibfnamefont {J.}~\bibnamefont
  {Handsteiner}}, \bibinfo {author} {\bibfnamefont {A.}~\bibnamefont
  {Hochrainer}}, \bibinfo {author} {\bibfnamefont {K.}~\bibnamefont {Phelan}},
  \bibinfo {author} {\bibfnamefont {F.}~\bibnamefont {Steinlechner}}, \bibinfo
  {author} {\bibfnamefont {J.}~\bibnamefont {Kofler}}, \bibinfo {author}
  {\bibfnamefont {J.-A.}\ \bibnamefont {Larsson}}, \bibinfo {author}
  {\bibfnamefont {C.}~\bibnamefont {Abell\'an}}, \bibinfo {author}
  {\bibfnamefont {W.}~\bibnamefont {Amaya}}, \bibinfo {author} {\bibfnamefont
  {V.}~\bibnamefont {Pruneri}}, \bibinfo {author} {\bibfnamefont {M.~W.}\
  \bibnamefont {Mitchell}}, \bibinfo {author} {\bibfnamefont {J.}~\bibnamefont
  {Beyer}}, \bibinfo {author} {\bibfnamefont {T.}~\bibnamefont {Gerrits}},
  \bibinfo {author} {\bibfnamefont {A.~E.}\ \bibnamefont {Lita}}, \bibinfo
  {author} {\bibfnamefont {L.~K.}\ \bibnamefont {Shalm}}, \bibinfo {author}
  {\bibfnamefont {S.~W.}\ \bibnamefont {Nam}}, \bibinfo {author} {\bibfnamefont
  {T.}~\bibnamefont {Scheidl}}, \bibinfo {author} {\bibfnamefont
  {R.}~\bibnamefont {Ursin}}, \bibinfo {author} {\bibfnamefont
  {B.}~\bibnamefont {Wittmann}}, \ and\ \bibinfo {author} {\bibfnamefont
  {A.}~\bibnamefont {Zeilinger}},\ }\href {\doibase
  10.1103/PhysRevLett.115.250401} {\bibfield  {journal} {\bibinfo  {journal}
  {Phys. Rev. Lett.}\ }\textbf {\bibinfo {volume} {115}},\ \bibinfo {pages}
  {250401} (\bibinfo {year} {2015})}\BibitemShut {NoStop}%
\bibitem [{\citenamefont {Rosenfeld}\ \emph {et~al.}(2017)\citenamefont
  {Rosenfeld}, \citenamefont {Burchardt}, \citenamefont {Garthoff},
  \citenamefont {Redeker}, \citenamefont {Ortegel}, \citenamefont {Rau},\ and\
  \citenamefont {Weinfurter}}]{rosenfeld2017}%
  \BibitemOpen
  \bibfield  {author} {\bibinfo {author} {\bibfnamefont {W.}~\bibnamefont
  {Rosenfeld}}, \bibinfo {author} {\bibfnamefont {D.}~\bibnamefont
  {Burchardt}}, \bibinfo {author} {\bibfnamefont {R.}~\bibnamefont {Garthoff}},
  \bibinfo {author} {\bibfnamefont {K.}~\bibnamefont {Redeker}}, \bibinfo
  {author} {\bibfnamefont {N.}~\bibnamefont {Ortegel}}, \bibinfo {author}
  {\bibfnamefont {M.}~\bibnamefont {Rau}}, \ and\ \bibinfo {author}
  {\bibfnamefont {H.}~\bibnamefont {Weinfurter}},\ }\href {\doibase
  10.1103/PhysRevLett.119.010402} {\bibfield  {journal} {\bibinfo  {journal}
  {Phys. Rev. Lett.}\ }\textbf {\bibinfo {volume} {119}},\ \bibinfo {pages}
  {010402} (\bibinfo {year} {2017})}\BibitemShut {NoStop}%
\bibitem [{\citenamefont {{\v{S}}upi{\'{c}}}\ and\ \citenamefont
  {Bowles}(2020)}]{supic2019}%
  \BibitemOpen
  \bibfield  {author} {\bibinfo {author} {\bibfnamefont {I.}~\bibnamefont
  {{\v{S}}upi{\'{c}}}}\ and\ \bibinfo {author} {\bibfnamefont {J.}~\bibnamefont
  {Bowles}},\ }\href {\doibase 10.22331/q-2020-09-30-337} {\bibfield  {journal}
  {\bibinfo  {journal} {{Quantum}}\ }\textbf {\bibinfo {volume} {4}},\ \bibinfo
  {pages} {337} (\bibinfo {year} {2020})}\BibitemShut {NoStop}%
\bibitem [{\citenamefont {Holz}\ \emph {et~al.}(2019)\citenamefont {Holz},
  \citenamefont {Kampermann},\ and\ \citenamefont {Bruß}}]{holz2019}%
  \BibitemOpen
  \bibfield  {author} {\bibinfo {author} {\bibfnamefont {T.}~\bibnamefont
  {Holz}}, \bibinfo {author} {\bibfnamefont {H.}~\bibnamefont {Kampermann}}, \
  and\ \bibinfo {author} {\bibfnamefont {D.}~\bibnamefont {Bruß}},\
  }\href@noop {} {\enquote {\bibinfo {title} {A genuine multipartite {B}ell
  inequality for device-independent conference key agreement},}\ } (\bibinfo
  {year} {2019}),\ \Eprint {http://arxiv.org/abs/1910.11360} {arXiv:1910.11360
  [quant-ph]} \BibitemShut {NoStop}%
\bibitem [{\citenamefont {Almeida}\ \emph {et~al.}(2010)\citenamefont
  {Almeida}, \citenamefont {Bancal}, \citenamefont {Brunner}, \citenamefont
  {Acín}, \citenamefont {Gisin},\ and\ \citenamefont {Pironio}}]{almeida2010}%
  \BibitemOpen
  \bibfield  {author} {\bibinfo {author} {\bibfnamefont {M.~L.}\ \bibnamefont
  {Almeida}}, \bibinfo {author} {\bibfnamefont {J.-D.}\ \bibnamefont {Bancal}},
  \bibinfo {author} {\bibfnamefont {N.}~\bibnamefont {Brunner}}, \bibinfo
  {author} {\bibfnamefont {A.}~\bibnamefont {Acín}}, \bibinfo {author}
  {\bibfnamefont {N.}~\bibnamefont {Gisin}}, \ and\ \bibinfo {author}
  {\bibfnamefont {S.}~\bibnamefont {Pironio}},\ }\href {\doibase
  10.1103/physrevlett.104.230404} {\bibfield  {journal} {\bibinfo  {journal}
  {Phys. Rev. Lett.}\ }\textbf {\bibinfo {volume} {104}},\ \bibinfo {pages}
  {230404} (\bibinfo {year} {2010})}\BibitemShut {NoStop}%
\bibitem [{\citenamefont {Peres}(1999)}]{peres1999}%
  \BibitemOpen
  \bibfield  {author} {\bibinfo {author} {\bibfnamefont {A.}~\bibnamefont
  {Peres}},\ }\href {\doibase 10.1023/A:1018816310000} {\bibfield  {journal}
  {\bibinfo  {journal} {Foundations of Physics}\ }\textbf {\bibinfo {volume}
  {29}},\ \bibinfo {pages} {589} (\bibinfo {year} {1999})}\BibitemShut
  {NoStop}%
\bibitem [{\citenamefont {Pitowsky}(1991)}]{pitowsky1991}%
  \BibitemOpen
  \bibfield  {author} {\bibinfo {author} {\bibfnamefont {I.}~\bibnamefont
  {Pitowsky}},\ }\href {\doibase 10.1007/BF01594946} {\bibfield  {journal}
  {\bibinfo  {journal} {Mathematical Programming}\ }\textbf {\bibinfo {volume}
  {50}},\ \bibinfo {pages} {395} (\bibinfo {year} {1991})}\BibitemShut
  {NoStop}%
\bibitem [{\citenamefont {Clauser}\ \emph {et~al.}(1969)\citenamefont
  {Clauser}, \citenamefont {Horne}, \citenamefont {Shimony},\ and\
  \citenamefont {Holt}}]{clauser1969}%
  \BibitemOpen
  \bibfield  {author} {\bibinfo {author} {\bibfnamefont {J.~F.}\ \bibnamefont
  {Clauser}}, \bibinfo {author} {\bibfnamefont {M.~A.}\ \bibnamefont {Horne}},
  \bibinfo {author} {\bibfnamefont {A.}~\bibnamefont {Shimony}}, \ and\
  \bibinfo {author} {\bibfnamefont {R.~A.}\ \bibnamefont {Holt}},\ }\href
  {\doibase 10.1103/PhysRevLett.23.880} {\bibfield  {journal} {\bibinfo
  {journal} {Phys. Rev. Lett.}\ }\textbf {\bibinfo {volume} {23}},\ \bibinfo
  {pages} {880} (\bibinfo {year} {1969})}\BibitemShut {NoStop}%
\bibitem [{\citenamefont {Clauser}\ \emph {et~al.}(1970)\citenamefont
  {Clauser}, \citenamefont {Horne}, \citenamefont {Shimony},\ and\
  \citenamefont {Holt}}]{clauser1970erratum}%
  \BibitemOpen
  \bibfield  {author} {\bibinfo {author} {\bibfnamefont {J.~F.}\ \bibnamefont
  {Clauser}}, \bibinfo {author} {\bibfnamefont {M.~A.}\ \bibnamefont {Horne}},
  \bibinfo {author} {\bibfnamefont {A.}~\bibnamefont {Shimony}}, \ and\
  \bibinfo {author} {\bibfnamefont {R.~A.}\ \bibnamefont {Holt}},\ }\href
  {\doibase 10.1103/PhysRevLett.24.549} {\bibfield  {journal} {\bibinfo
  {journal} {Phys. Rev. Lett.}\ }\textbf {\bibinfo {volume} {24}},\ \bibinfo
  {pages} {549} (\bibinfo {year} {1970})}\BibitemShut {NoStop}%
\bibitem [{\citenamefont {Froissart}(1981)}]{froissart1981}%
  \BibitemOpen
  \bibfield  {author} {\bibinfo {author} {\bibfnamefont {M.}~\bibnamefont
  {Froissart}},\ }\href {\doibase 10.1007/BF02903286} {\bibfield  {journal}
  {\bibinfo  {journal} {Il Nuovo Cimento B (1971-1996)}\ }\textbf {\bibinfo
  {volume} {64}},\ \bibinfo {pages} {241} (\bibinfo {year} {1981})}\BibitemShut
  {NoStop}%
\bibitem [{\citenamefont {Śliwa}(2003)}]{sliwa2003}%
  \BibitemOpen
  \bibfield  {author} {\bibinfo {author} {\bibfnamefont {C.}~\bibnamefont
  {Śliwa}},\ }\href {\doibase https://doi.org/10.1016/S0375-9601(03)01115-0}
  {\bibfield  {journal} {\bibinfo  {journal} {Phys. Lett. A}\ }\textbf
  {\bibinfo {volume} {317}},\ \bibinfo {pages} {165 } (\bibinfo {year}
  {2003})}\BibitemShut {NoStop}%
\bibitem [{\citenamefont {Collins}\ and\ \citenamefont
  {Gisin}(2004)}]{collins2004}%
  \BibitemOpen
  \bibfield  {author} {\bibinfo {author} {\bibfnamefont {D.}~\bibnamefont
  {Collins}}\ and\ \bibinfo {author} {\bibfnamefont {N.}~\bibnamefont
  {Gisin}},\ }\href {\doibase 10.1088/0305-4470/37/5/021} {\bibfield  {journal}
  {\bibinfo  {journal} {J. Phys. A: Math. Gen.}\ }\textbf {\bibinfo {volume}
  {37}},\ \bibinfo {pages} {1775} (\bibinfo {year} {2004})}\BibitemShut
  {NoStop}%
\bibitem [{\citenamefont {Bernards}\ and\ \citenamefont
  {G\"uhne}(2020)}]{bernards2020}%
  \BibitemOpen
  \bibfield  {author} {\bibinfo {author} {\bibfnamefont {F.}~\bibnamefont
  {Bernards}}\ and\ \bibinfo {author} {\bibfnamefont {O.}~\bibnamefont
  {G\"uhne}},\ }\href {\doibase 10.1103/PhysRevLett.125.200401} {\bibfield
  {journal} {\bibinfo  {journal} {Phys. Rev. Lett.}\ }\textbf {\bibinfo
  {volume} {125}},\ \bibinfo {pages} {200401} (\bibinfo {year}
  {2020})}\BibitemShut {NoStop}%
\bibitem [{\citenamefont {Ziegler}(1995)}]{ziegler}%
  \BibitemOpen
  \bibfield  {author} {\bibinfo {author} {\bibfnamefont {G.~M.}\ \bibnamefont
  {Ziegler}},\ }\href@noop {} {\emph {\bibinfo {title} {Lectures on Polytopes:
  Updated Seventh Printing of the First Edition}}},\ Graduate Texts in
  Mathematics 152\ (\bibinfo  {publisher} {Springer-Verlag New York},\ \bibinfo
  {year} {1995})\BibitemShut {NoStop}%
\bibitem [{\citenamefont {Pironio}(2005)}]{pironio2005}%
  \BibitemOpen
  \bibfield  {author} {\bibinfo {author} {\bibfnamefont {S.}~\bibnamefont
  {Pironio}},\ }\href {\doibase 10.1063/1.1928727} {\bibfield  {journal}
  {\bibinfo  {journal} {Journal of Mathematical Physics}\ }\textbf {\bibinfo
  {volume} {46}},\ \bibinfo {pages} {062112} (\bibinfo {year}
  {2005})}\BibitemShut {NoStop}%
\bibitem [{\citenamefont {Navascu\'es}\ \emph {et~al.}(2007)\citenamefont
  {Navascu\'es}, \citenamefont {Pironio},\ and\ \citenamefont
  {Ac\'{\i}n}}]{navascues2007}%
  \BibitemOpen
  \bibfield  {author} {\bibinfo {author} {\bibfnamefont {M.}~\bibnamefont
  {Navascu\'es}}, \bibinfo {author} {\bibfnamefont {S.}~\bibnamefont
  {Pironio}}, \ and\ \bibinfo {author} {\bibfnamefont {A.}~\bibnamefont
  {Ac\'{\i}n}},\ }\href {\doibase 10.1103/PhysRevLett.98.010401} {\bibfield
  {journal} {\bibinfo  {journal} {Phys. Rev. Lett.}\ }\textbf {\bibinfo
  {volume} {98}},\ \bibinfo {pages} {010401} (\bibinfo {year}
  {2007})}\BibitemShut {NoStop}%
\bibitem [{\citenamefont {Wittek}(2015)}]{wittek2015}%
  \BibitemOpen
  \bibfield  {author} {\bibinfo {author} {\bibfnamefont {P.}~\bibnamefont
  {Wittek}},\ }\href {\doibase 10.1145/2699464} {\bibfield  {journal} {\bibinfo
   {journal} {ACM Trans. Math. Softw.}\ }\textbf {\bibinfo {volume} {41}}
  (\bibinfo {year} {2015}),\ 10.1145/2699464}\BibitemShut {NoStop}%
\bibitem [{\citenamefont {P\'al}\ and\ \citenamefont
  {V\'ertesi}(2010)}]{pal2010}%
  \BibitemOpen
  \bibfield  {author} {\bibinfo {author} {\bibfnamefont {K.~F.}\ \bibnamefont
  {P\'al}}\ and\ \bibinfo {author} {\bibfnamefont {T.}~\bibnamefont
  {V\'ertesi}},\ }\href {\doibase 10.1103/PhysRevA.82.022116} {\bibfield
  {journal} {\bibinfo  {journal} {Phys. Rev. A}\ }\textbf {\bibinfo {volume}
  {82}},\ \bibinfo {pages} {022116} (\bibinfo {year} {2010})}\BibitemShut
  {NoStop}%
\end{thebibliography}%

\end{document}